\documentclass[12pt]{article}

\usepackage[labelfont=bf,textfont=rm]{caption}

\usepackage{url}

\usepackage{rotating}

\usepackage[]{algorithm2e}
\usepackage{hyperref}
\usepackage{etoolbox} % needed for the \forcsvlist command
\usepackage{graphicx}
\usepackage{stmaryrd} %for \llbracket and \rrbracket
\usepackage{tcolorbox}
\usepackage{amssymb,amsmath,amsthm,amsfonts,pifont} %mathabx changes arrow heads (\rightarrow) \not to a vertical line (\not can be recovered with the \changenotsign command (see below)), but I didn't like the arrowheads, and it also makes \vee and \wedge smaller!

\usepackage[margin=1in]{geometry}

\renewcommand{\phi}{\varphi}

\newtheorem{theorem}{Theorem}

\newtheorem{corollary}{Corollary}
\newtheorem{lemma}{Lemma}

\theoremstyle{definition}
\newtheorem{example}{Example}

\usepackage{caption}

\newcommand{\set}[1]{\ensuremath{\left\{#1\right\}}}
\newcommand{\sem}[1]{\ensuremath{\llbracket#1\rrbracket}}

\usepackage{soul} %for striking through \st{asdf}

\newcommand{\I}[1]{\mathit{#1}}
\newcommand{\R}[1]{\textrm{#1}}
\newcommand{\mc}[1]{\mathcal{#1}}
\newcommand{\mt}{\mapsto}

\newcommand{\meet}{\sqcap}

\newcommand{\la}{\leftarrow}

\newcommand{\ra}{\rightarrow}

\newcommand{\asyn}{\ensuremath{\rightarrow}}

\newcommand{\rf}[1]{(\ref{#1})} %\ref with parentheses

\newcommand{\el}{v}

\newcommand\Bi{\begin{itemize}}
\newcommand\Ei{\end{itemize}}
\newcommand\Be{\begin{enumerate}}
\newcommand\Ee{\end{enumerate}}

\makeatletter
\newcommand{\items}[1]{
  \Bi\item{#1}\itrest}
\newcommand{\itrest}{\@ifnextchar\bgroup{\itrestp}{\Ei}}
\newcommand{\itrestp}[1]{\item{#1}\@ifnextchar\bgroup{\itrestp}{\Ei}}
\makeatother

\makeatletter

\newcommand{\enrest}{\@ifnextchar\bgroup{\enrestp}{\Ee}}
\newcommand{\enrestp}[1]{\item{#1}\@ifnextchar\bgroup{\enrestp}{\Ee}}
\makeatother

\newcommand\Bt{\begin{tabular}}
\newcommand\Et{\end{tabular}}

\newcommand{\cell}[1]{#1&}

\newcommand{\cells}[1]{%
  \forcsvlist{\cell}{#1}%
}

\makeatletter

\newcommand{\tabrest}{\@ifnextchar\bgroup{\tabrestp}{\Et}}
\newcommand{\tabrestp}[2]{\\ \cells{#1}#2\@ifnextchar\bgroup{\tabrestp}{\Et}}
\makeatother

\begin{document}

\title{\bf Bridging Abstract Dialectical Argumentation\\ and Boolean Gene Regulation}

\author{Eugenio Azpeitia\thanks{Centro de Ciencias Matemáticas, Universidad Nacional Autónoma de México; email: \texttt{eazpeitia@matmor.unam.mx}}
  \and
  Stan Mu\~noz Guti\'errez\thanks{Institute of Software Technology, Graz University of Technology; email: \texttt{sgutierr@ist.tugraz.at}}
  \and
  David A.\ Rosenblueth\thanks{Instituto de Investigaciones en Matem\'aticas Aplicadas y en Sistemas, Universidad Nacional Aut\'onoma de M\'exico; email: \texttt{drosenbl@unam.mx}}
  \and
  Octavio Zapata\thanks{Instituto de Matem\'aticas, Universidad Nacional Aut\'onoma de M\'exico; email: \texttt{octavio@im.unam.mx}}
  }

\maketitle

\begin{abstract}
This paper leans on two similar areas so far detached from each other.
On the one hand,
Dung's pioneering contributions to \emph{abstract} argumentation, almost thirty years ago,
gave rise to a plethora of successors,
including \emph{abstract dialectical frameworks} (ADFs).
On the other hand, \emph{Boolean networks} (BNs),
devised as models of gene regulation,
have been successful for studying the behavior of molecular processes within cells.
ADFs and BNs are similar to each other:
both can be viewed as functions from vectors of bits to vectors of bits.
As soon as similarities emerge between these two formalisms, however, differences appear.
For example,
conflict-freedom is prominent in argumentation (where we are interested in a self-consistent, i.e., conflict-free, set of beliefs) but absent in BNs.
By contrast, asynchrony (where only one gene is updated at a time) is conspicuous in BNs and lacking in argumentation.
Finally, while a monotonicity-based notion occurs in signed reasoning of both argumentation and gene regulation, a different, derivative-based notion only appears in the BN literature.
To identify common mathematical structure between both formalisms, these differences need clarification.
This contribution is a partial review of both these areas, where we cover enough ground to exhibit their more evident similarities, to then reconcile some of their
apparent
differences.
We highlight a range of avenues of research resulting from ironing out discrepancies between these two fields.
Unveiling their common concerns should enable these two areas to cross-fertilize so as to transfer ideas and results between each other.
\end{abstract}

\section{Introduction}
Argumentation theory studies reasoning as performed by a collection of agents.
\emph{Abstract argumentation}, in particular, dispenses with the structure of arguments, producing a simple, elegant formalism.
At the same time,
gene networks model biological interaction of combined gene activity.
\emph{Boolean gene networks}, specifically, represent the activation level of a gene with only two possible values, yielding an approach located in the abstract end of the spectrum.
Our objective is to draw attention to the common mathematical structure between abstract argumentation and Boolean gene networks.
We review both areas, show commonalities, and reconcile some of their differences.

\paragraph{Argumentation.}
The study of reasoning has been, at least since Ancient Greece, closely linked with logic.
Deductive inference in classical logic thus became to be regarded as the ideal correct reasoning.
Argumentation, by contrast, had been studied as a branch of rhetoric, and hence separate from logic.
There has been a recent shift, however, incorporating argumentation as stemming from communication between reasoners.
A critique of the standard position began in philosophy, with Toulmin's pioneering book~\cite{toulmin-58}.
His contribution is twofold.
First, he discarded the traditionally assumed one-to-one mapping between natural-language statements and logical formulas, and
noticed %in his pioneering book~\cite{toulmin-58}
the importance of speech communication in reasoning.
He also
studied arguments employed by goal-directed agents in support of their objectives~\cite{woods-21}.
His work influenced first the informal-logic community and next gave rise to multiple developments~\cite{woods-21}.

Argumentation has also become a relevant subject in psychology.
It is interesting that humans have numerous cognitive biases, such as the confirmation bias or the sunk-cost fallacy.
As a way to explain these biases, researchers have proposed that human reasoning properly belongs to discussions, persuasion, and argument exchange~\cite{mercier-sperber-11,mercier-sperber-17}.

Another field where abstract argumentation has had a significant impact is legal reasoning.
In a legal case, decisions are made and conclusions arrived at through
conflict resolution and deliberation. %"conflict-resolution arguments" Woods, Reasoning and Argumentation, p 72 par 2 left
Consequently, legal applications lend themselves to be modeled with formal argumentation techniques.
There is hence a large body of literature in this area~\cite{bench-capon-dunne-06,prakken-21}.

Although some precursory ideas~\cite{pollock-87,pollock-91,pollock-92} had already been published, Dung's 1995 article~\cite{dung-95} is usually credited as giving birth to the field of abstract argumentation.
In this contribution, we thus begin by summarizing Dung's work.
Dung's argumentation frameworks, however,
only consider counterarguments: the sole relation between arguments is that of ``attack''.
Hence, our goal is, instead, to give an account of abstract ``dialectical'' frameworks~\cite{brewka-woltran-10}, having more general forms of argument interaction than Dung's % ``attack''
relation.
Abstract argumentation has branched off in numerous directions and the bibliography of this area is vast.
Rather than summarizing the state of the art in abstract argumentation, our purpose is to cover enough material to reveal similarities and differences with Boolean gene networks.
Hence, we limit ourselves to finite frameworks and omit connections with logic programming.
We refer the reader to~\cite{baroni-gabbay-giacomin-van-der-torre-18,bench-capon-dunne-07} for overviews of abstract-argumentation research.

\paragraph{Boolean gene networks.}
von Neumann %Richard & Ruet "From kernels..."
marked~\cite{von-neumann-66} the advent of Boolean networks, further developed~\cite{kauffman-69} by Kauffman as models of gene interaction.
Kauffman's original formalism models the temporal changes in the activity of the genes due to their regulatory interactions.
States are vectors of gene values and all genes are updated synchronously (i.e., simultaneously); stable states correspond to the observed gene activity of cellular types.
Since Robert's pivotal contributions~\cite{robert-78,robert-80,robert-86}, %synchronous
Kauffman's original Boolean networks have been % extended
modified with asynchrony (where only one gene is updated at a time) by Thomas~\cite{thomas-d-ari-90}.
More recently, Paulev\'e has put forth ``most permisive'' Boolean networks~\cite{pauleve-kolcak-chatain-haar-20},
able to reproduce all behaviors of a corresponding quantitative model.

A Boolean model of a cell is often built using information from \emph{two} experimental sources: (a) interactions between genes (where often each interaction is reported in a single article) and (b) a set of expected set of stable states (corresponding to a set of cellular types).
Usually each gene interaction has an associated sign, denoting either an activation or a repression.

Boolean networks are now not only an important model of gene regulation~\cite{abou-jaoude-et-al-16} but relevant in their own right~\cite{akutsu-18,gaucherel-thero-puiseux-bonhomme-17,grandi-lorini-perrussel-15,montagud-et-al-22,poindron-21,rozum-campbell-newby-nasrollahi-albert-23}.
Although usually called Boolean \emph{gene} network, this formalism is also used for modeling signal transduction, metabolic pathways~\cite{albert-thakar-li-zhang-albert-08}, and ecological processes~\cite{nasrollahi-campbell-albert-23}, cellular processes~\cite{bhattacharyya-et-al-21}, and opinion diffusion~\cite{grandi-lorini-perrussel-15}.
Reviews are:~\cite{abou-jaoude-et-al-16,albert-thakar-14,le-novere-15,naldi-et-al-15,samaga-klamt-13,schwab-kuhlwein-ikonomi-kuhl-kestler-20,wang-saadatpour-albert-12}.

\paragraph{Shared notions.}
It is immediately apparent that there must be commonalities underlying abstract argumentation frameworks and Boolean networks.
Both an abstract dialectical framework and a Boolean network can be viewed as a function from vectors of bits to vectors of bits.
A more scrupulous study, however, must address a number of questions suggesting that perhaps performing such an amalgamation is more subtle than it appears at first sight.

Let us identify some shared notions first.
Argumentation's acceptance conditions play the role of Boolean network's regulatory functions.
An abstract dialectical framework minus its acceptance conditions corresponds to the interaction graph of a Boolean network.
Two-valued interpretations are states.
Three-valued interpretations~\cite{brewka-ellmauthaler-strass-wallner-woltran-13}, in turn, have been developed in the Boolean-network literature under several names: partial states~\cite{irons-06}, symbolic states~\cite{klarner-14,klarner-bockmayr-siebert-14,klarner-bockmayr-siebert-15,klarner-siebert-15}, and ``states of a subset of the inputs''~\cite{zanudo-albert-13}.
Furthermore, admissible sets of arguments have been also distinguished in the context of Boolean networks %by these authors as,
as ``hypercubes''~\cite{moon-lee-pauleve-23},``trap spaces''~\cite{klarner-14} and ``quasi-attractors''~\cite{zanudo-albert-13};
complete extensions are %parallel
symbolic steady states.
Even more remarkably,
\emph{preferred extensions are minimal trap spaces}.

\paragraph{Differences.}
These commonalities are significant in view of abstract dialectical frameworks and Boolean networks having divergent goals. %, however.
Ideally, argumentation considers a framework solved when a single solution is found, comprising of a set of self-consistent set of arguments attacking other arguments.
In general, however, Dung frameworks may not have a single solution, so that work in abstract argumentation usually discriminates between skeptical and credulous reasoning, depending on whether an argument is accepted in all extensions or one extension of a particular semantics.

Gene-regulation modeling, by contrast, identifies each solution with a different cell type, and therefore accepts several solutions.
More generally, argumentation is typically interested in a static solution, whereas gene regulation, having evolved (as a research field) from dynamical systems, incorporates \emph{time} in its models.

Another significant difference is the absence of the concept of \emph{conflict-freedom} in
Boolean gene regulation.
Moreover, already when transiting from Dung's abstract argumentation frameworks to abstract dialectical frameworks, conflict-freedom seems to have disappeared.
(To be sure, several notions called conflict-freeness have been introduced in abstract dialectical frameworks, but these notions play a different role from that of Dung's formalism as they do not appear in the definition of admissibility.)
This is puzzling as often abstract dialectical frameworks are referred to as a generalization of Dung's formalism.

This brings us to the question of whether or not abstract dialectical frameworks are a generalization of Dung's formalism.
We observe that the semantics developed by~\cite{brewka-ellmauthaler-strass-wallner-woltran-13} for abstract dialectical frameworks coincides with Caminada's~\cite{caminada-08} semantics for Dung frameworks, when an abstract dialectical framework is also a Dung framework.
We also show that an explicit treatment of
conflict-freedom is disposed of %already implicit
in Caminada's approach, and therefore so it is in abstract dialectical frameworks (where it has been reintroduced) and in Boolean networks.

By comparison, after Thomas's work~\cite{thomas-d-ari-90}, \emph{asynchrony}
pervades much of the gene Boolean-network literature, whereas it does not seem to have a corresponding notion in argumentation.
It is easier to see argumentation as related to \emph{synchronous} Boolean networks.
We explain away this discrepancy.
We observe that
although the ultimate interest in some of the Boolean-network contributions may be in asynchronous update, the features shared with argumentation
are independent of the update discipline.
Moreover, we remark that the asynchrony of gene Boolean networks does occur in opinion diffusion viewed as a generalization of argumentation.

\paragraph{Computation.}
Once we have dealt with essential similarities and differences between these fields, we turn our attention to computation.
On the argumentation side, we summarize an algorithm~\cite{strass-wallner-15} inspired by complexity analysis in abstract dialectical frameworks.
An implementation is the tool \texttt{k++ADF}~\cite{linsbichler-maratea-niskanen-wallner-woltran-18}, which repetitively calls a Boolean satisfiability (SAT) solver.
On the Boolean gene network side, we cover
\texttt{tsconj}~\cite{trinh-benhamou-pastva-soliman-24}, which uses answer set programming.

One tool can compute notions of the opposite camp: \texttt{k++ADF} can obtain maximal trap spaces and \texttt{tsconj} can obtain preferred interpretations.
It is not clear, however, how to compute notions such as the ``naive'' semantics of argumentation or ``cyclic'' attractors with tools of the opposite camp.

\paragraph{Signed reasoning.}
Further commonship and disagreement between abstract argumentation and Boolean networks appears with signed reasoning.
Arbitrary abstract dialectical frameworks are sometimes restricted~\cite{brewka-woltran-10} to having only ``supporting'' and ``attacking'' acceptance conditions, based on \emph{monotonicity}.
These are ``bipolar'' networks.

This notion also occurs in Boolean regulation, where it has its roots back to Waddington in 1942, who developed the concept of ``canalization''~\cite{waddington-42} to explain why the wild type of an organism %Waddington, p 563, right, par -1
(i.e., the form occurring in nature under the influence of natural selection) is less variable than most mutants.
(The name canalization alludes to robustness.) %or buffering.)
Kauffman interpreted~\cite{kauffman-71,kauffman-74},\cite[pp.\ 203, 204]{kauffman-93} this concept in the context of Boolean networks, giving a biological basis~\cite{raeymaekers-02} to a subset of all possible Boolean networks.
Moreover, his interpretation coincides~\cite{li-adeyeye-murrugarra-aguilar-laubenbacher-13} with monotonicity and hence with supporting and attacking acceptance conditions.
(There is also a large body of literature in digital-circuit theory, where this notion is called unate function, e.g.,~\cite{balogh-dong-lidicky-mani-zhao-23,mcnaughton-61}~\cite[p 61]{o-donell-14}.)
Works employing this notion of interaction in Boolean networks are:~\cite{akutsu-melkman-tamura-yamamoto-11,aracena-08,aracena-richard-salinas-17,benevs-brim-kadlecaj-pastva-safranek-20,li-adeyeye-murrugarra-aguilar-laubenbacher-13,moon-lee-pauleve-23,mori-mochizuki-17,veliz-cuba-11}.

Nevertheless, signed reasoning in Boolean networks is the source of another discrepancy between dialectical argumentation and Boolean regulation.
We devote our attention to a different kind of interaction also occurring in the Boolean-network literature,
based on the \emph{discretization of the time derivative}~\cite[p 47]{o-donell-14} instead of on monotonicity.
Boolean gene networks can be viewed as an abstraction of more faithful models of gene regulation, employing differential equations.
Thus, another way of giving a biological basis to certain Boolean networks is to abstract to a Boolean version the time derivative.
The sign taken from the time derivative can be justified as an approximation of the Hill function~\cite[p 9, 10]{alon-19}.
Works using the time-derivative gene interaction are:~\cite{aracena-cabrera-crot-salinas-21,comet-et-al-13,munoz-carrillo-azpeitia-rosenblueth-18,naldi-thieffry-chaouiya-07,richard-19,richard-rossignol-comet-bernot-guespin-michel-merieau-12,richard-ruet-13}.

It is natural to wonder whether or not both notions are equivalent to each other.
We observe that it is possible to reconcile them.
First we remark that the monotonicity-based interaction is the negation of the derivative-based interaction of the opposite sign.
We next note that by explicitly constraining the derivative-based interaction to be \emph{strict},
both notions coincide.
It is not necessary, however, to explicitly take strict monotonicity-based interactions.
Moreover, when non-strict interactions are used, the possibility of expressing \emph{optional} interactions arises.

\paragraph{Structure of this contribution.}
The rest of this contribution is organized as follows.
We
devote Sect.~\ref{sect-arg} to argumentation frameworks, where we explain away conflict-freedom and cover abstract dialectical frameworks.
Sect.~\ref{sect-bool} deals with Boolean networks, observing that in spite of stemming from a time-related concern (attractors), trap spaces do not depend on the update discipline.
Sect.~\ref{sect-adm-trap} summarizes prominent computer tools for computation of admissible interpretations/trap spaces and reports results of our tests on an argumentation tool and a Boolean-network tool.
We cover signed reasoning in Sect.~\ref{sect-signed}, where we delve into the differences between
monotonicity-based and derivative-based signs.
Concluding remarks are in Sect.~\ref{sect-concl}.

\section{Argumentation frameworks}
\label{sect-arg}
In this section, we first review Dung's abstract argumentation frameworks~\cite{dung-95}.
Next, we give Caminada labelings~\cite{caminada-06a} for Dung frameworks, which serve as a stepping stone to abstract dialectical frameworks~\cite{brewka-woltran-10}.
We subsequently summarize the semantics developed~\cite{brewka-ellmauthaler-strass-wallner-woltran-13} by Brewka, Ellmauthaler, Strass, Waller, and Woltran for abstract dialectical frameworks, and show that their characteristic operator can be viewed as a generalization of Caminada's version of the update function.

\subsection{Dung's abstract argumentation frameworks}

\paragraph{Dung framework.}
A \emph{Dung framework} consists of a finite set $A$ whose elements are called \emph{arguments},  and a set $R$ of ordered pairs of arguments.   %and a relation $R\subseteq A\times A$ called \emph{attack}. 
If $(a,b)\in R$, then we say that the argument $a$ \emph{attacks} the argument $b$, or that $b$ is an \emph{attacker} of $a$, and we represent this graphically by drawing a directed edge (i.e., an arrow) from $a$ to $b$. 

\begin{example}
\label{ex-running}
Consider the following example, of a child trying to persuade her mother that she may play. 
The set $A = \set{a,b,c,d}$ contains the following arguments:  
  \begin{itemize}
  \item[($d$)] mother: ``You cannot play because it's late.''
  \item[($c$)] Eloise: ``But it's only 19h00.''
  \item[($b$)] mother: ``You cannot play until you've done your homework.''
  \item[($a$)] Eloise: ``I've finished my homework.''
  \end{itemize}
 The set $R=\set{(a,b), (b,c), (c,d), (d,c)}$ is the relation  of attack between the elements of $A$. 
 This example can be described by the arrow diagram in Figure~\ref{fig-running}.
 \end{example}

\begin{figure}[ht]
\begin{center}
\includegraphics[scale=0.8]{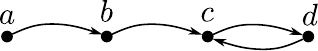}
\end{center}
\caption{\label{fig-running} Graphical representation of Example~\ref{ex-running}.} %Kopka & Daly, p 178
\end{figure}

 \paragraph{Notation.}
 %Let $\mc{D} = (A,R)$ be a Dung framework. 
If $\mc{D} = (A,R)$ is a Dung framework and $a\in A$, then we write 
\[
\R{par}(a) = \{b\in A \mid  (b,a)\in R \}
\]
 %and call the elements of this 
 for the set of attackers of $a$. 
In the context of dialectic argumentation (see Sect.~\ref{sect-adf}), the elements of $\R{par}(a)$ are called the \emph{parents} of $a$. 
For each subset  $X\subseteq A$,  we write  %$X^{\rightarrow}$
\[
R(X) = \{b\in A \mid (a,b) \in R \textrm{ for some } a \in X\}
\] 
for the set of arguments which are \emph{attacked by} $X$.  %$X^{\leftarrow}$
We write $A\setminus X$ for the set-theoretic complement of $X$, that is, for the set of all elements of $A$ that do not belong to $X$:
\[
A\setminus X = \{b\in A \mid b\notin X\}.
\]

 \paragraph{Stable extension.}
 We say that $X\subseteq A$ is  \emph{conflict-free} (also called \emph{independent set}~\cite{neumann-lara-71}) if for no arguments $a$ and $b$ in $X$ we have that $a$ attacks $b$.
In other words, $X$ is conflict-free if the elements of $X$ only attack arguments that belong to the complement of $X$. 
Since $A\setminus R(X)$ is the set of arguments not attacked by $X$, it follows that $X$  is conflict-free if and only if   $X\subseteq A\setminus R(X)$.
If each argument in $A\setminus X$ is attacked by $X$, then
we say that $X$ is \emph{emitting}. 
In symbols, $X$ is emitting if $A\setminus X \subseteq R(X)$. 
Thus $X$ is emitting if and only if $X \supseteq A\setminus R(X)$. 
A \emph{stable extension} is a subset of arguments which is conflict-free and emitting.  
Hence, a necessary and sufficient condition that $X$ be a stable extension is that 
\begin{equation}
  X = A\setminus R(X). \label{stable}
\end{equation}
An (impressionistic) example of a stable extension is displayed in Figure~\ref{fig-stable-extension}.

\begin{figure}[ht]
\begin{center}
\includegraphics[scale=0.5]{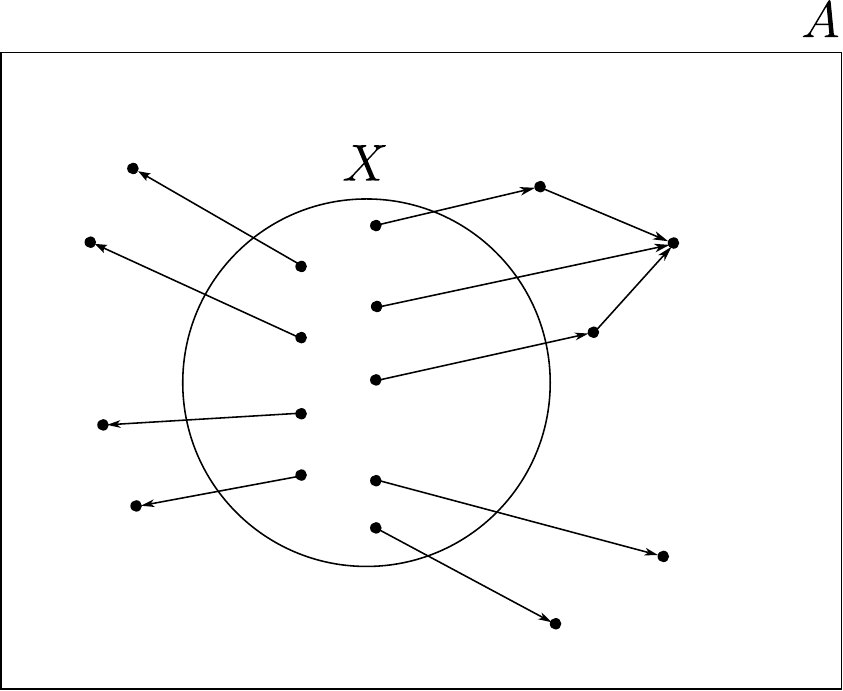}
\end{center}
\caption{\label{fig-stable-extension} Example of a stable extension.
All arguments in $X$ attack all arguments not in $X$.}
\end{figure}

 \paragraph{Subsets versus two-valued functions.}
Let $\{0,1\}^{A}$ be the set of all functions from $A$ into $\set{0,1}$. 
The elements of $\{0,1\}^{A}$ will be called \emph{two-valued functions} on $A$. 
The constant functions $0\in\{0,1\}^{A}$ and $1\in\{0,1\}^{A}$ are defined, for each $a\in A$, by $0(a)=0$ and $1(a)=1$.
Each subset $X\subseteq A$ is associated with a two-valued  function $x \in \set{0,1}^A$, called the \emph{indicator function} of $X$, defined by %defined for each $a\in A$ by %: A \to \{0,1\}
\[
x(a) = 
\begin{cases}
1	& \R{if } \, a\in X,\\
0	& \R{if } \, a\notin X.
\end{cases}
\]
 The correspondence that maps each subset to its indicator function is a one-to-one correspondence from the set of all subsets of $A$ onto  the set  of two-valued functions on $A$. %$\{0,1\}^{A}$. 
The inverse correspondence maps each two-valued  function $y \in\set{0,1}^A$ to its \emph{support}, which is the set 
\[
Y = \{a\in A \mid y(a)=1\}.
\]
The support of the constant zero function is the empty set $\emptyset$, which has no elements, and the support of the constant 1 function is the whole set $A$.

There is partial order $\leq$ on the set $\set{0,1}^A$ defined pointwise: $x\leq y$ if $x(a)\leq y(a)$ for all $a\in A$.
For us, the most important thing to note here is that if $X$ and $Y$ are subsets of $A$, and $x$ and $y$ are their indicator functions, 
then $X\subseteq Y$ if and only if $x\leq y$.

\paragraph{Dung's labeling.}
A necessary and sufficient condition that $X$ be a stable extension (see, e.g.,~\cite{bezem-grabmayer-walicki-12}) is that its indicator function $x\in \{0,1\}^{A}$ be such that
\begin{equation}
\label{eq-dung}
x(a) = 
\begin{cases}
1&\R{if } \,  x(b)=0 \, \R{ for all $b\in \R{par}(a)$,}\\
0&\R{if } \,  x(b)=1 \, \R{ for some $b\in \R{par}(a)$.}
\end{cases}
\end{equation}
Intuitively, $b \in \R{par}(a)$ is labeled with $1$ (i.e., is in $X$), then $a$ must be labeled with $0$ (i.e., is not in $X$).
Hence, there are no arrows in $X$; i.e., $X$ is conflict-free.
At the same time, the only possibility for an argument $a$ be labeled with $0$ is if there is an arrow from $b$ to $a$ and $b$ is labeled with $1$.
Thus, all arguments not in $X$ are attacked by $X$.
Note also that if an argument is not attacked at all (e.g., argument $a$ in Example~\ref{ex-running}), then it is labeled with $1$ by the indicator function.

\paragraph{Two-valued update function.}
The right-hand side of equation~\eqref{stable} can be regarded as a function that maps each set $X$ to  the set $A\setminus R(X)$.
This function was first studied by Neumann-Lara in 1971~\cite{neumann-lara-71} (in the context of kernel theory, where the edges are reversed), %Teorema 2, definicion de B
then by Pollock~\cite{pollock-87,pollock-91,pollock-92}, and more recently by Dung~\cite{dung-95}, Grossi~\cite{grossi-12}, and Strass~\cite{strass-13}. %Grossi, fixedpoints
Equation~\eqref{stable} denotes the fixed points of this function.
Thus, its fixed points %of this function
are, by definition, the stable extensions. 

It will be more convenient to use the indicator functions of sets of arguments as opposed to the sets themselves.
The \emph{two-valued update function} of $\mathcal{D} = (A,R)$ is the mapping $f_\mathcal{D} : \{0,1\}^A \to \{0,1\}^A$ defined for each $x\in \{0,1\}^A$ and $a\in A$ by 
\begin{equation}
\label{two-update}
f_\mathcal{D}(x)(a) = 
\begin{cases}
	1&\R{if } \,  x(b)=0 \, \R{ for all $b\in \R{par}(a)$,}\\
	0&\R{if } \,  x(b)=1 \, \R{ for some $b\in \R{par}(a)$.}
\end{cases}
\end{equation}

If $p$ and $q$ are elements of the set $\{0,1\}$, we write $p\wedge q$ for the least element of the set $\{p,q\}$, and $p\vee q$ for the greatest element of $\{p,q\}$. 
More generally, if $\{p_i\}_{i\in I}$ is a family of elements $p_i$ of $\{0,1\}$, indexed by the elements $i$ of an arbitrary finite set $I$, the least element of this family is denoted by $\bigwedge_{i \in I} p_i$  and the greatest element by  $\bigvee_{i \in I} p_i$. 
We write $\neg p$ for the complementary value of each $p$ in $\{0,1\}$ (i.e., $\neg p = 1- p$). 
In our context, the operations $\land$, $\lor$ and $\neg$ are called \emph{conjunction}, \emph{disjunction} and \emph{negation}, respectively.
This leads us to another definition of $f_\mathcal{D}$:
\begin{equation}
\label{eq-nor}
f_\mathcal{D}(x)(a) \, = \bigwedge_{b\in \R{par}(a)} \neg x(b) \, = \, \neg\big( \bigvee_{b\in \R{par}(a)} x(b)\big).
\end{equation}

It is possible for a Dung framework to have no stable extensions.
Examples are:
\begin{center}
  \includegraphics[scale=0.8]{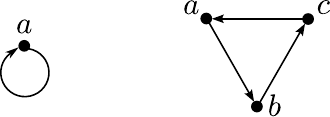}
\end{center}
For instance, in the self-loop example, if we assume that $x(a)=1$, then $x(a)=0$ and vice versa.
This situation is usually considered anomalous in abstract argumentation and is closely connected with logical paradoxes~\cite{cook-04,dyrkolbotn-12,dyrkolbotn-walicki-14}.

\paragraph{Admissibility.}
A first way to rectify this 
situation is to weaken the notion of emission. 
Instead of requiring $X$ to attack all other arguments,
we now require $X$ to attack \emph{all arguments that attack $X$}. %Neumann-Lara
This was done by Neumann-Lara~\cite{neumann-lara-71}
and Dung~\cite{dung-95}. 
An \emph{admissible set} is a subset of arguments $X$ such that 
$X$ is conflict-free  and all the attackers of each argument in $X$ are attacked by $X$. 
In other words,  $X$ is admissible if it is conflict-free and $\R{par}(a) \subseteq R(X)$ for each $a\in X$.
An example of an admissible set is displayed in Figure~\ref{fig-admissible}.

\begin{figure}[ht]
\begin{center}
\includegraphics[scale=0.6]{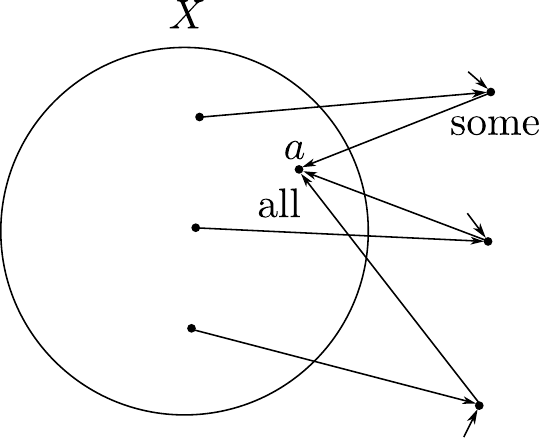}
\end{center}
\caption{\label{fig-admissible} Example of an admissible set. All arguments in $X$ attack all arguments that attack $X$.}
\end{figure}

The fact that two consecutive edges are involved in the notion of admissibility suggests a double application of $f_\mc{D}$~\cite{dung-95}.
We thus apply $f_\mc{D}$ to both sides of~\eqref{two-update}, so that the right-hand side of $f_\mc{D}(f_\mc{D}(x))(a)$ has $f_\mc{D}(x)(b)$ instead of $x(b)$. 
Hence
\[
 f_\mc{D}(f_\mc{D}(x))(a) = 
 \begin{cases}
 	          1&\R{if } \,  f_\mc{D}(x)(b)=0 \, \R{ for all $b\in \R{par}(a)$,}\\
	          0&\R{if } \,   f_\mc{D}(x)(b)=1 \, \R{ for some $b\in \R{par}(a)$.}
                  % 0&\R{if }\exists b\in A .\, bRa \land f_\mc{D}(x)(b)=1,\\
                  % 1&\R{if }\forall b\in A .\, bRa \Ra \underline{f_\mc{D}(x)(b)=0}.\\
 \end{cases}
\]
 ``Unfolding''  this application of $f_D$, we get that $f_\mc{D}(f_\mc{D}(x))(a) = 1$ if and only if  for each $b\in \R{par}(a)$ there exists a $c\in \R{par}(b)$ such that $x(c) = 1$. 
Suppose now that the subsets $X$ and $Y$ are the supports of the functions $x$ and $y = f_\mc{D}(f_\mc{D}(x))$.
 Then $a\in Y$ if and only if $\R{par}(a) \subseteq R(X)$; this implies that  $\R{par}(a) \subseteq R(X)$ for each $a\in Y$.
Therefore, $X$ is an admissible set if and only if $X$ is conflict-free and $X\subseteq Y$.

Let $x$ be the indicator function of a conflict-free subset.
The observations of the previous paragraph imply that a necessary and sufficient condition for the support of $x$ to be  an admissible set is that $x \leq f_\mc{D}(f_\mc{D}(x))$. 
We are interested in the admissible sets whose indicator function are the fixed points of $f_\mc{D}\circ f_\mc{D}$ (i.e., the double application of $f_\mc{D}$).

The \emph{characteristic function} of $\mc{D} = (A, R)$ is the mapping $F_\mc{D} : \{0,1\}^A \to \{0,1\}^A$ defined for each $x\in \{0,1\}^A$ and $a\in A$ by 
\[
F_\mc{D}(x)(a) = (f_\mc{D}\circ f_\mc{D}) (x) (a)= f_\mc{D}(f_\mc{D}(x))(a).
\]
A \emph{complete extension} is a conflict-free set of arguments $X$ whose indicator function $x$ is a fixed point of the characteristic function (i.e., $F_\mc{D}(x)(a) = x(a)$ for all $a\in A$).

As opposed to $f_\mc{D}$, the characteristic function $F_\mc{D}$ has a fixed point for all Dung frameworks.
A special case of a well-known result, often called the Knaster--Tarski theorem~\cite{tarski-55}, asserts that if an operator $G$ from a complete lattice into itself is monotone, then $G$ has a fixed point.
(An operator $G$ is monotone if $X\leq Y$ implies $G(X)\leq G(Y)$;
$G$ is anti-monotone if $X\leq Y$ implies $G(X)\leq G(Y)$.)
It is easy to see that  the set of two-valued functions $\set{0,1}^A$ is a complete lattice, and that $F_\mc{D}$ is a monotone operator on this complete lattice.
Hence, such a fixed point exists.

Another important property of $F_\mc{D}$ is that it preserves conflict-freedom when applied to a conflict-free set of arguments (see, e.g.~\cite[Lemma 18]{dung-95}).

Let $n$ be the number of elements in the finite set $A$.
If  $A = \set{a_1,\dots,a_n}$, then we represent each two-valued function $x\in\set{0,1}^A$ as a vector $x=(x_1,\dots,x_n)\in \set{0,1}^n$ with $x_i = x(a_i)$ for all $i=1,\dots,n$. 
We consider each such vector as the binary expansion of a nonnegative integer less than or equal to $n-1$. 
In this way, we associate each two-valued function with a decimal numeral. %whose binary expansion is the corresponding 
Figure~\ref{fig-running-state-transition-graph} shows the two-valued update function of Example~\ref{ex-running}.
In decimal numerals, the stable extensions of $\mc{D}$ are $9$ and $10$; the complete extensions are $8$, $9$, and $10$ ($11$ is not because it is not conflict-free).
  
\begin{figure}[ht]
\begin{center}
\includegraphics[scale=0.8]{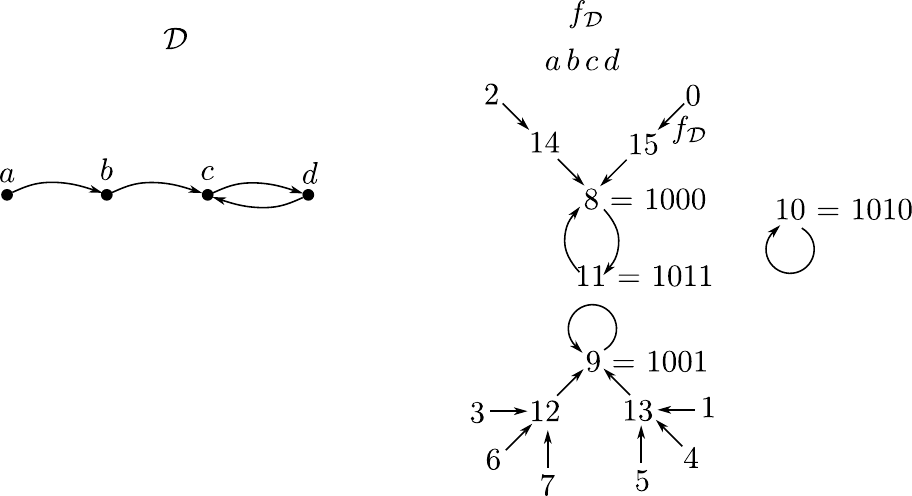}
\end{center}
\caption{\label{fig-running-state-transition-graph} Graphical representation of Example~\ref{ex-running} (left) along with its two-valued update function $f_\mc{D}$ (right), where each arrow of the latter is an application of $f_\mc{D}$.}
\end{figure}

\paragraph{Other extensions.}
Complete extensions always exist, but are not unique.
Therefore, additional properties are sometimes required.
For example, the $\subseteq$-minimal
complete extension of $F_\mc{D}$, called the \emph{grounded extension}, is unique.
The grounded extension can be obtained by repeatedly applying $F_\mc{D}$ starting from the empty set ($0$ in $\set{0,1}^A$), which is vacuously conflict-free.
(The grounded extension of $\mc{D}$ is $8$.)
A disadvantage of the grounded extension is that
the set of accepted arguments is minimal; hence sometimes it is considered too conservative for argumentation problems.
At the other extreme, \emph{preferred extensions} are $\subseteq$-maximal complete extensions (in the sense that no proper superset of the extension is also complete). %Grossi "fixed points", Bondarenko et al "An abstract argumentation..." (but I changed admissible to complete)
Hence, they are not necessarily unique.
(The preferred extensions of $\mc{D}$ are $9$ and $10$.)
A \emph{semi-stable extension}~\cite{caminada-carnielli-dunne-06,caminada-carnielli-dunne-11}
is defined as %an admissible set
a complete extension where $X\cup R(X)$ is $\subseteq$-maximal.
An advantage is that they always exist, and if the Dung framework has a stable extension, then the semi-stable extensions coincide with the stable extensions.

Other extensions, that are not necessarily complete, are the \emph{naive}~\cite{bondarenko-dung-kowalski-toni-97} and \emph{stage}~\cite{verheij-96} extensions.
A naive extension is a $\subseteq$-maximal conﬂict-free set. %taken from Strass "Approximating..."
A stage extension is a pair $(\I{defeated},\I{undefeated})$ of sets of arguments where any argument in $\I{defeated}$ is attacked by $\I{defeated}$ and $\I{defeated}\cup\I{undefeated}$ is $\subseteq$-maximal.

\paragraph{Caminada's labeling.}
Caminada solved the possible absence of stable extensions in a different way.
A reason why a Dung framework may not have stable extensions is that Dung's labeling may not be well defined.
Caminada added a third value $u$
covering precisely such cases.
Thus, Caminada's approach can be viewed as an attempt to always have stable extensions through three-valued functions.
The result is an alternative definition of complete extensions.

Let $\{0,1,u\}^{A}$ be the set of all functions from $A$ into $\set{0,1,u}$. 
The elements of $\{0,1,u\}^{A}$ will be called \emph{three-valued interpretations} of $\mc{D}$. 
We say that $v\in \{0,1,u\}^{A}$ 
           %            \colon A\ra\set{0,1,u}$
is a \emph{Caminada labeling} of $\mc{D}=(A,R)$ if
\begin{equation}
\label{eq-caminada}
\el(a) = 
\begin{cases}
   1&\R{if } \,  \el(b)=0 \, \R{ for all $b\in \R{par}(a)$,}\\
   0&\R{if } \,  \el(b)=1 \, \R{ for some $b\in \R{par}(a)$,}\\
   u&\R{otherwise}.
\end{cases}
\end{equation}
which is an elaboration of Dung's labeling~\eqref{eq-dung}.
Thus, $\el(a)=u$ precisely when
$\el(b) \neq 0$ for some $b\in \R{par}(a)$ and $\el(b) \neq 1$ for all $b\in \R{par}(a)$. 
Equivalently, of the attackers of $a$, none are labeled with $1$, and at least one is labeled with $u$.

At first sight, Caminada labelings do not seem to be an improvement compared with Dung's, because there is a one-to-one correspondence with Dung's complete extensions.
An advantage, as we will see, is that instead of two applications of the update function, there is only one. 

First, we observe that Caminada's labelings are conflict-free.
\begin{lemma}
\label{lemma-conflict-freedom}
If $\el$ is a Caminada labeling of $\mc{D}$, then the set $\set{a\in A\mid\el(a)=1}$  is conflict-free.
\end{lemma}

Mappings from and to Caminada labelings and complete extensions can be established~\cite{caminada-06a} as follows.
Let $\mc{L} : \set{0,1}^A \to \set{0,1,u}^A$ be the mapping %from
defined for each $x\in \set{0,1}^A$ and $a\in A$ by 
\[
  \mc{L}(x)(a) 
  =
  \begin{cases}
  1 &\R{if } \, x(a)=1,\\
  0 & \R{if } \, \R{par}(a) \neq \emptyset,\\
  u &\R{otherwise}.
  \end{cases}
\]

\begin{lemma}[\cite{caminada-06a}]
If $x$ is the indicator function of a complete extension,
then the function $\mc{L}(x)$ in $\set{0,1,u}^A$
is a Caminada labeling of $\mc{D}$.  
\end{lemma}

\begin{figure}
\begin{center}
\includegraphics[scale=0.8]{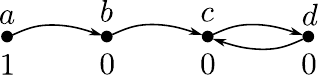}
\end{center}
\begin{align*}
\mc{L}(x)(a) & = 1\\
\mc{L}(x)(b) & = 0\\
\mc{L}(x)(c) & = u\\
\mc{L}(x)(d) & = u
\end{align*}
\caption{\label{fig-caminada} The Caminada labeling resulting from applying $\mc{L}$ to the complete extension named $8$ of Figure~\ref{fig-running-state-transition-graph}.}
\end{figure}

Similarly, let $\mc{M} : \set{0,1,u}^A \to \set{0,1}^A$ be the mapping
defined for each $v\in \set{0,1,u}^A$ and $a\in A$ by
\[
  \mc{M}(v)(a) = 
  \begin{cases}
                    1 & \R{if } \, v(a)=1,\\
                    0 & \R{otherwise}.
\end{cases}                    
\]

\begin{lemma}[\cite{caminada-06a}]
If $\el$ is a Caminada labeling of $\mc{D}$,
then the support of the two-valued function $\mc{M}(\el)$ is a complete extension.
\end{lemma}

\paragraph{Three-valued update function.}

We can define a binary operation $\meet$ on the set $\set{0,1,u}$, which can be read as \emph{consensus}, given by $0 \meet 0 = 0$, $1 \meet 1 = 1$, and $u$ otherwise. 
For each $p,q\in \set{0,1,u}$, we write $p\leq_i q$ ($i$ for information) if $p\meet q = p$. 
The relation $\leq_i$ is a partial order on  $\set{0,1,u}$, and with respect to this order we have that  $u<_i0$ and $u<_i1$. 
If $v$ and $w$ are elements of the set $\set{0,1,u}^A$, we write $v \leq_i w$ if $v(a)\leq_i w(a)$ for all $a\in A$. 
It can be shown that $\set{0,1,u}^A$ is a complete meet-semilattice with respect to $\leq_i$. 
Figure~\ref{fig-meet-semilattice} depicts a meet-semilattice of three-valued interpretations.
\begin{figure}[ht]
\begin{center}
   \includegraphics[scale=0.8]{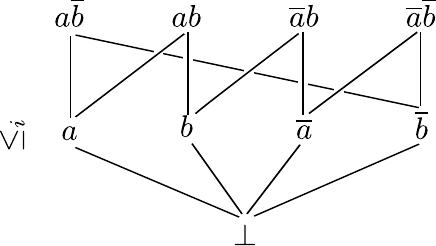}
\end{center}
\caption{\label{fig-meet-semilattice} A meet-semilattice of three-valued interpretations. 
%The notation 
%$a\overline{b}$ abbreviates $\set{a\mt 1, b\mt 0}$,
%$a$ abbreviates $\set{a\mt 1, b\mt u}$, and
%$\bot$ abbreviates $\set{a\mt u, b\mt u}$. 
%In this case, $A=\set{a,b}$. 
We write $a\overline{b}$ for the three-valued interpretation $v$ on $A=\set{a,b}$ given by  $v(a) =1$ and $v(b)=0$. 
Similarly, $a$ denotes the interpretation $v(a)= 1$ and $v(b) = u$, and $\bot$ the interpretation $v(a)= v(b) = u$.
}
\end{figure}

Just as we obtained the two-valued update function based on Dung's labeling, we can get an update function from Caminada's labeling. 
The \emph{three-valued update function} of $\mathcal{D} = (A,R)$ is the mapping $g_\mathcal{D} : \{0,1,u\}^A \to \{0,1,u\}^A$ defined for each $v\in \{0,1,u\}^A$ and $a\in A$ by 
\begin{equation*}
\label{three-update}
g_\mathcal{D}(v)(a) = 
\begin{cases}
	1&\R{if } \,  v(b)=0 \, \R{ for all $b\in \R{par}(a)$,}\\
	0&\R{if } \,  v(b)=1 \, \R{ for some $b\in \R{par}(a)$,}\\
	u&\R{otherwise.}
\end{cases}
\end{equation*}

It can be proved that $g_\mc{D}$ is monotone with respect to $\leq_i$, and hence 
that $g_\mc{D}$ has a fixed point for any Dung framework $\mc{D}$.

\subsection{Abstract dialectical frameworks}
\label{sect-adf}
Abstract dialectical frameworks were developed~\cite{brewka-woltran-10} by Brewka and Woltran as an extension of Dung's argumentation frameworks~\cite{dung-95}, by allowing arbitrary Boolean functions over the set of parents of an argument,  as opposed to the conjunction of the negation of the parents~\rf{eq-nor}.

Before giving the official definition, we have to make some preliminary observations. 
We say that a Boolean function $h : \set{0,1}^A \ra\set{0,1}$ \emph{depends} on $b \in A$ if 
$h(x) \neq h(y)$ for some $x,y\in\{0,1\}^A$ which only differ in $b$ in the sense that $x(b)\neq y(b)$ and $x(a) = y(a)$ for all $a\neq b$. 
For example, the function $h(x) = (x(a)\land x(b))\lor(\neg x(a)\land x(b))$ depends on $b$ because $h(x)=0$ if $x(a)=0$ and $x(b)=0$, and $h(y)=1$ if $y(a)=0$ and $y(b)=1$.
Recall that the symbols $\land$, $\lor$ and $\neg$ denote the operations of conjunction, disjunction and negation, which are defined on the set $\set{0,1}$ by $p\land q = \min\set{p,q}$, $p\lor q = \max\set{p,q}$ and $\lnot p = 1- p$. 

Each Boolean function can be represented as a propositional formula. 
A propositional formula is an expression built up from variables using the symbols $\land$, $\lor$ and $\neg$,
and using parentheses. 
In our case, the variables shall be elements of the set $A$. 
Every propositional formula can be evaluated by giving an assignment of the values 0 and 1 to the variables. 
We write $x(\phi)$ for the result of the evaluation of the propositional formula $\phi$ with the two-valued function $x\in\set{0,1}^A$.
 For example, if $\phi$ is the propositional formula $(a\land b)\lor(\neg a\land b)$, then $x(\phi) = (x(a)\land x(b))\lor(\neg x(a)\land x(b))$.  
 In this example, it is immediate to see that  the formula $\phi$ represents the Boolean function $h$ given in the previous paragraph in the sense that $h(x) = x(\phi)$ for any $x\in \set{0,1}^A$. 
This representation, however, is not unique. 
For instance,  the function $h$ may well be represented by the formula $(a\land\neg a)\lor b$, and also by the formula $b$.

\paragraph{Abstract dialectical framework.}
An \emph{abstract dialectical framework} (ADF) consists of a finite set $A$ of arguments, a set $R$ of ordered pairs of arguments, and a family of functions $C=\{C_a\}_{a\in A}$ such that $C_a : \set{0,1}^{\R{par}(a)}\ra\set{0,1}$ depends on each $b\in\R{par}(a)$. %such that 
The function $C_a$ is called the \emph{acceptance condition} of $a$.

It is convenient to represent each acceptance condition $C_a$ as a propositional formula $\phi_a$, and so we may think of $C$ as a family propositional formulas $\set{\phi_a}_{a\in A}$.

% Any two-valued interpretation $x : A \ra \set{0,1}$ can be extended to all propositional formulas
% as follows:
% \begin{align*}
% & x(\top) = 1\\
% & x(\neg\phi) = 1 && \hspace{-100pt} \R{ if } x(\phi)=0\\
% & x(\phi\land\psi) = 1 && \hspace{-100pt} \R{ if } x(\phi)=1 \R{ and } x(\psi)=1
% \end{align*}
% We view other operators as shorthands (i.e., $\phi\lor\psi$ is the same as $\lnot(\lnot\phi\land\lnot\psi)$, etc.). 

For each $v\in \set{0,1,u}^A$,  let $\sem{v}$ be the set of all two-valued functions $x\in \set{0,1}^A$ satisfying $x(a)=v(a)$ for all $a\in A$ such that $v(a)\neq u$. 
In other words: 
\[
\sem{v} = \set{x\in\set{0,1}^A \mid v\leq_i x}.
\]
The elements of $\sem{v}$ are called the \emph{completions} of $v$.

If $\mc{D} = (A, R, C)$ is an ADF, the \emph{characteristic operator} of $\mc{D}$ is the mapping $\Gamma_\mc{D} : \{0,1,u\}^A \to \{0,1,u\}^A$ defined for each $x\in \{0,1\}^A$ and $a\in A$ by 
\begin{equation}
\label{eq-characteristic}
  \Gamma_\mc{D}(v)(a) =
\begin{cases}
      1 & \R{if } \, x(\phi_a) = 1 \, \R{ for all $x\in \sem{v}$,} \\
      0 & \R{if } \, x(\phi_a) = 0 \, \R{ for all $x\in \sem{v}$,}\\
      u & \R{otherwise}.
  \end{cases}
\end{equation}

\begin{lemma}[\cite{brewka-woltran-10}] 
The operator $\Gamma_\mc{D}$ is monotone with respect to $\leq_i$.
\end{lemma}
This property of $\Gamma_\mc{D}$ allows us to define several semantics~\cite{brewka-ellmauthaler-strass-wallner-woltran-13} for ADF.
The \emph{grounded interpretation} for $\mc{D}$
is the least fixpoint of $\Gamma_\mc{D}$.
A three-valued interpretation $v$ for an ADF $\mc{D}$ is \emph{admissible} if $v\leq_i\Gamma_\mc{D}(v)$.
A three-valued interpretation $v$ for an ADF $\mc{D}$ is \emph{complete} if $\Gamma_\mc{D}(v)=v$ (note the absence of the notion of conflict-freedom).
A three-valued interpretation $v$ for an ADF $\mc{D}$ is \emph{preferred} if $v$ is $\leq_i$-maximal admissible.

\paragraph{ADFs as generalization of Dung frameworks.}
ADFs are a generalization, not of Dung's original formalism, but of Caminada's variant.
Indeed, if $\mc{D}=(A,R)$ is the Dung framework given by $A = \set{a,b}$ and $R=\set{(a,b),(b,a)}$, 
the support of the two-valued function $x\in\set{0,1}^A$ given by $x(a)=x(b)=0$ is a complete extension of $\mc{D}$, but $x$ is not a complete interpretation under the ADF semantics of~\cite{brewka-ellmauthaler-strass-wallner-woltran-13}.

\begin{lemma}
Let $\mc{D} = (A,R,C)$ be an ADF and suppose $\mc{D}$ is also a Dung framework.
Then $\Gamma_\mc{D}=g_\mc{D}$.
\end{lemma}
\begin{proof}
We must instantiate the characteristic operator~\rf{eq-characteristic} to the negation of the disjunction of the parents of each argument (equation~\rf{eq-nor}).
Let $v\in\set{0,1,u}^A$ and $a\in A$.
Clearly, $\Gamma_\mc{D}(v)(a)=g_\mc{D}(x)(a)$ for all $x\in\sem{v}$ if $v(b)\in\set{0,1}$ for all $b\in\R{par}(a)$.
So, we must only consider the cases in which $v(b)=u$.
Recall next that $g_\mc{D}(a)=u$ when, of the parents of $a$, none are labeled with $1$, and at least one is labeled with $u$.
Then $\Gamma_\mc{D}(v)(a)=u$.
The reason is that for some $x\in\sem{v}$, $x(b)=0$
(take $x\in\sem{v}$ such that $x(b)=0$ for all $b\in\R{par}(a)$, in which case $x(\phi_b)=1$)
and for some $x\in\sem{v}$, $x(b)=1$
(in which case $x(\phi_b)=0$).
This coincides with $g_\mc{D}(x)(a)$.
Otherwise, there is a parent of $a$ labeled with $1$ and another one labeled with $u$, in which case $\Gamma_\mc{D}(v)(a)=0$.
This also coincides with $g_\mc{D}(x)(a)$.
\end{proof}

\paragraph{Conflict-freedom.}
Several notions called conflict-freedom (or conflict-freeness) have been incorporated~\cite{brewka-woltran-10,strass-13,strass-wallner-15} in abstract dialectic argumentation. %strass-wallner-15, p 35
The %last and
simplest one~\cite{linsbichler-maratea-niskanen-wallner-woltran-18,strass-wallner-15} is a weakening of admissibility: a three-valued interpretation $v$ as \emph{conflict-free} if, for all $a\in A$,
$v(a)=p$ implies that there exists $x\in\sem{v}$ such that $x(\phi_a)=p$, for $p\in\set{0,1}$.

\paragraph{Complexity.}
Deciding whether a given three-valued interpretation is admissible is coNP-complete;
deciding whether it is preferred is coNP$^{\R{NP}}$-complete~\cite{strass-wallner-15}.
(We refer the reader to~\cite{papadimitriou-94} for notions of computational complexity.)

\section{Boolean networks}
\label{sect-bool}

\paragraph{Stable states.}
Kauffman introduced~\cite{kauffman-69} synchronous Boolean networks as models of cells, where
variables have only two values, and time is discrete.
Each gene has an associated regulatory function, and all genes are updated simultaneously.
Each state of a cell in which the cell does not change must correspond to a state of the network in which the network does not change either.
Equivalently, each cellular type must correspond to a fixed point, or single-state attractor, of the update function of the network.

Boolean networks are one of myriad gene-network models, which include stochastic models, differential equations, piecewise-linear equations, and discrete models~\cite{de-jong-02}. %, and Boolean networks~\cite{de-jong-02}.

\paragraph{Multiple-state attractors.}
Attractors of size greater than one often represent periodic or oscillatory behavior~\cite{albert-thakar-14}.
There are many examples of oscillatory behaviors in biological systems, including the cell cycle, circadian rhythms, neural oscillations, the production of new organs in plants, and biochemical oscillations~\cite{pavlidis-12}.
Moreover, recent works have suggested that cyclic attractors could also be important for the robustness and adaptability of biological systems (e.g.,~\cite{de-leon-vazquez-jimenez-matadamas-guzman-resendis-antonio-22,ordaz-arias-diaz-alvarez-zuniga-martinez-sanchez-balderas-martinez-22}).
Boolean networks have been capable of reproducing oscillatory behavior (e.g.,~\cite{akman-watterson-parton-binns-millar-ghazal-12,faure-thieffry-09,ge-qian-09}).

% A Boolean network is often defined as a function $f : \set{0,1}^n\ra\set{0,1}^n$. 
% We prefer the alternative definition of Boolean network, as it coincides with that of ADF. %given in Definition~\ref{def:adf}. 

\paragraph{Boolean network.}
A \emph{Boolean network} consists of a finite set $A$ whose elements are called \emph{genes}, a set $R$ of ordered pairs of genes, and a family of Boolean functions $C=\{C_a\}_{a\in A}$ such that $C_a : \set{0,1}^{\R{par}(a)}\ra\set{0,1}$ depends on each $b\in\R{par}(a)$. 
The function $C_a$ is called the \emph{regulatory function} of $a$.

As with ADFs, it is common to represent the regulatory functions $\{C_a\}_{a\in A}$ as propositional formulas
$\set{\phi_a}_{a\in A}$.

\paragraph{Interaction graph.}
If $\mc{B}=(A,R,C)$ is a Boolean network, then the pair $(A,R)$ is a directed graph called the \emph{interaction graph} of $\mc{B}$.
The interaction graph is an important notion in Boolean networks.
Starting with Thomas's conjectures~\cite{thomas-81},
properties of the state-transition graph are often inferred from the interaction graph, especially when signed reasoning is incorporated~\cite{aracena-08,mori-mochizuki-17,pauleve-richard-12,remy-ruet-08,remy-ruet-thieffry-08,richard-10,richard-comet-07}.

\paragraph{Asynchronous networks.}
Thomas modified~\cite{thomas-d-ari-90} Kauffman's formalism into a model forbidding update simultaneity, where a single gene is updated at a time.
A state has one successor for each gene that is updated.
Which gene is updated is undetermined.

\paragraph{State-transition graphs.}
We associate a \emph{state-transition graph} $(\set{0,1}^A,\ra)$ with a Boolean network $\mc{B}=(A,R,\set{\phi_a})$ together with an \emph{update discipline} as follows:
\items
{For the \emph{synchronous} update discipline, $x\ra y$ if $y(a)=x(\phi_a)$ for all $a\in A$. %[Klarner's thesis, p 27] [Trinh et al. Scalable enumeration, p 3]
That is, all arguments are updated simultaneously.
}
{For the \emph{asynchronous} update discipline, $x\asyn y$ if  either $y(a)=x(\phi_a)$ for all $a\in A$ and $y=x$, or               there is an $a\in A$ such that  $y(a)=x(\phi_a)$ and $y(b)=x(b)$ for all  $b\neq a$.
Hence, $y$ differs from $x$ in the value of either zero arguments or exactly one argument.
Sometimes state-transition graphs of asynchronous Boolean networks do not have self-loops (i.e., the first case in the definition is omitted; e.g.,~\cite{chatain-haar-pauleve-18}).
Some authors, however, include such loops, making $\ra$ a total (i.e., serial) relation, e.g.,~\cite{klarner-14}.
} 
A common generalization of synchronous and asynchronous disciplines is often considered.
``Most permissive'' Boolean networks~\cite{pauleve-kolcak-chatain-haar-20} are a more recent development, able to reproduce all behaviors of a corresponding quantitative model.

\paragraph{Trap spaces.}
Klarner devised~\cite{klarner-14} trap spaces,
which are related to ``seeds''~\cite{siebert-11} and ``stable motifs''~\cite{zanudo-albert-13}, as a means to approximate attractors.

An \emph{attractor} with respect to $\ra$ is 
a $\subseteq$-minimal
set of two-valued functions $E\subseteq \set{0,1}^A$ such that for each $x\in E$, if $x\ra y$ then $y\in E$. 
A \emph{trap space} with respect to $\ra$ is a three-valued interpretation $v\in \set{0,1,u}^A$
such that for each any $x\in\sem{v}$, $x\ra y$ implies $y\in\sem{v}$.
In other words, $v$ is a trap space if and only if $\sem{v}$ is an attractor.

\paragraph{Join-semilattice.}
In the Boolean-network literature, the same meet-semilattice as that for ADFs is at times employed~\cite{klarner-14}.
However, the inverse join-semilattice has become a de facto standard in Boolean networks.
Although in principle it is immaterial which of these two semilattices is used, we will now change to the join-semilattice for Boolean networks so as to conform with the literature (we use $\leq_k$, where $k$ is for knowledge).
Thus, a maximal preferred interpretation corresponds to a minimal trap space.
Figure~\ref{fig-join-semilattice} depicts a join-semilattice of three-valued interpretations.
\begin{figure}[ht]
\begin{center}
   \includegraphics[scale=0.8]{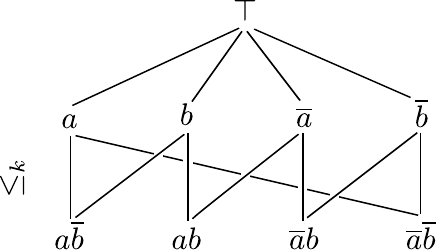}
\end{center}
\caption{\label{fig-join-semilattice} A join-semilattice of three-valued interpretations.}
\end{figure}

\begin{theorem}[\cite{klarner-14}]
\label{thm-adm-trap}
Let $\mc{D}=(A,R,C)$ be an ADF, and let $v\in\set{0,1,u}^A$.
Then $v$ is an admissible interpretation if and only if $\sem{v}$ is a trap space.
\end{theorem}

The proof of Theorem~\ref{thm-adm-trap} does not depend on the update discipline.
Hence, we have:

\begin{corollary}[\cite{klarner-14,trinh-benhamou-soliman-23}] 
The set of trap spaces of a Boolean network is the same regardless of the update discipline.
\end{corollary}

A minimal trap space contains at least one attractor regardless of the update discipline.
Moreover, any two minimal trap spaces are disjoint.
Hence, the set of minimal trap spaces is a good approximations to the set of attractors~\cite{klarner-streck-siebert-17,trinh-benhamou-pastva-soliman-24}.

Note that minimal trap spaces coincide with preferred interpretations.
\begin{theorem}
Let $\mc{D}$ be an ADF.
The set of preferred interpretations of $\mc{D}$ with respect to $\leq_i$ is equal to the set of minimal trap spaces of $\mc{D}$ with respect to $\leq_k$.
\end{theorem}

\paragraph{Complexity.}
Independently from~\cite{strass-wallner-15}, the same complexity results as for ADFs have been derived for Boolean networks:
Deciding whether a three-valued interpretation is a trap space is coNP-complete, whereas deciding whether it is a minimal trap space is coNP$^{\R{NP}}$-complete~\cite{moon-lee-pauleve-23}.

\section{Computation of admissible interpretations/trap\\spaces}
\label{sect-adm-trap}

Having identified connections between abstract dialectical argumentation and Boolean networks, it is now natural to turn to algorithms for the computation of admissible interpretations and trap spaces.
Both fields have followed parallel developments in that there have been efforts to discover methods for computing these notions in as large as possible instances.
From the point of view of one area, the principles, techniques, and ideas of the other area become available.

A direct approach to computing admissible interpretations/trap spaces given an ADF/BN $\mc{D}$ would compute $\Gamma_\mc{D}$ given a three-valued interpretation $v\in \set{0,1,u}^A$ by evaluating the acceptance conditions/regulatory functions with an exponential number of two-valued functions, namely with each completion $x\in\sem{v}$. % each completion $x\in\sem{v}$ of $v$ on $\phi_a$,
Shunning this exponential blow-up is of paramount importance.

At first sight, it might appear that a three-valued logic could be used to avoid this blow-up, such as Kleene's strong logic~\cite{kleene-52}.
Nevertheless,~\cite{baumann-heinrich-20} showed that no truth-functional three-valued logic covers $\Gamma_\mc{D}$.
(A logic is truth-functional if the evaluation of
a composed formula only depends of the truth values of its constituting subformulas.)

We now mention some techniques used in tools for computing admissible interpretations.
A SAT solver~\cite{biere-heule-van-maaren-walsh-09} is a computer program that attempts to solve the SAT problem, i.e., given a propositional formula in conjunctive normal form, decide whether or not there is a variable assignment making the formula true.

In Answer Set Programming (ASP)~\cite{brewka-eiter-truszczynski-11,eiter-ianni-krennwallner-09,gebser-kaminski-kaufmann-schaub-22,lifschitz-19}, a (disjunctive) ``program'' consists of a
set of rules of the form:
\[
A_1\lor\cdots\lor A_\ell \la B_1\land\cdots\land B_m\land\R{not}\, C_1\land\cdots\land\R{not}\, C_n 
\]
where we write $\phi\la \psi$ for the propositional formula $\phi\lor\neg \psi$, and $\R{not}\, \phi$ can be interpreted as negation as failure. 
Essentially, a rule means that any of $A_i$ is true if all of $B_j$ are true and none of $C_k$ can be proved to be true.
An ASP solver computes ``stable models'' or ``answer sets'', which are solutions to the program.

An ordered, reduced, binary-decision diagram (BDD)~\cite{bryant-92,bryant-86,huth-ryan-04,meinel-theobald-1998,wegener-00} is an often concise data structure representing a Boolean function.
It is a directed acyclic graph with each internal node having exactly two children and with two leaves.
The root corresponds to the function,
each internal node to a ``subfunction'' (i.e., a cofactor),
and the leaves respectively to 0 and 1.
It can also be seen as a representation of a set
or of a finite-state automaton.

Petri nets~\cite{peterson-77} are a graphical model of information flow with two types of nodes: ``places'' and ``transitions''.
Places are connected by arrows to transitions, and transitions to places.
The flow of information is controlled by the movement of ``tokens'', which cause a transition to fire when all of its input places have a token.
A connection between Petri nets and Boolean networks was first established in~\cite{chaouiya-remy-ruet-thieffry-04}.

\subsection{Computation of admissible interpretations in ADFs}
Here we skim through prominent ADF tools, but slow down for \texttt{k++ADF}.

The \texttt{DIAMOND} systems~\cite{ellmauthaler-strass-13,ellmauthaler-strass-14,strass-ellmauthaler-17} use ASP.
This family of tools a priori computes truth tables of the acceptance conditions\footnote{See \url{https://sourceforge.net/projects/diamond-adf/}, file \texttt{transform.pl}.}
Admissible interpretations are found by repetitively generating interpretations and checking whether or not they satisfy the definition of such interpretations. %Ellmauthaler & Strass "The Diamond system for argumentation. Preliminary report", 2013
The grounded model is directly computed by repetitively applying the characteristic operator.
Preferred interpretations are computed with the \texttt{\#minimize} directive.

\texttt{QADF}~\cite{diller-wallner-woltran-15} represents the different problems occurring in ADFs as quantified Boolean formulas (an extension of Boolean logic with quantification over propositional variables) and uses an off-the-shelf solver %(\texttt{DepQBF})
for quantified Boolean formulas.
Each problem corresponds to a single call to the solver. %Brewka, Diller et al. "Solving Advanced..."

\texttt{YADF}~\cite{brewka-diller-heissenberger-linsbichler-woltran-20} uses an approach called dynamic (where the \texttt{Y} comes from), because the encodings are generated individually for every instance.

\texttt{ADF-BDD}~\cite{ellmauthaler-gaggl-rusovac-wallner-22} employs BDDs as encondings to exploit the fact that many otherwise hard decision problems have polynomial time complexity~\cite{darwiche-marquis-02} when so encoded.
(At the time of writing this contribution, however, the computation of preferred interpretations with this tool had not been described~\cite{ellmauthaler-gaggl-rusovac-wallner-22} or implemented~\cite{ellmauthaler-23}.)

\texttt{k++ADF}~\cite{linsbichler-maratea-niskanen-wallner-woltran-18}, based on a SAT solver, exhibited the best performance in~\cite{linsbichler-maratea-niskanen-wallner-woltran-18} for many instances, compared with the other tools we have summed up.
The SAT solver must be such that in case of satisfiability, it returns an assignment satisfying the formula.
(\texttt{k++ADF} is written in \texttt{C++}, which might give it its name.)

At the heart of \texttt{k++ADF} there is an encoding of admissibility.
It happens that the negation of the definition of admissibility is in conjunctive normal form, and hence amenable to be given to a SAT solver.

A three-valued interpretation $v\in\set{0,1,u}^A$ is admissible if and only if the formula
\[
\Phi = \bigwedge_{v(a)=1} a\ \land \bigwedge_{v(a)=0} \neg a\ \to \bigwedge_{v(a)=1} \phi_a\  \land \bigwedge_{v(a)=0} \neg \phi_a
\]
is a tautology. 
Equivalently,  $v$ is admissible if $\neg\Phi$ is not satisfiable.
Hence, $\neg\Phi$ is used to discard candidate interpretations that are not admissible.

We summarize the computation of preferred interpretations.
Essentially, \texttt{k++ADF} constructs a conjunction of subformulas starting from an encoding of conflict-freedom (in the sense of Subsect.~\ref{sect-adf}), hence a priori discarding all non-conflict-free interpretations.
The formula $\Phi$ is used to test whether or not a candidate interpretation is admissible.
Once a candidate is rejected for not being admissible, a subformula encoding all the strictly greater interpretations than a given interpretation is used to generate another candidate (which next candidate is generated depends on the solver).
The last of such interpretations is preferred.
The preferred interpretation so obtained is blocked by adding its negation to the formula given to the solver so as to generate a different preferred interpretation.

\subsection{Computation of  trap spaces in Boolean networks}

We mention some important Boolean-network tools, and go deeper into \texttt{mpbn} and \texttt{tsconj}.

\texttt{PyBoolNet}~\cite{klarner-14,klarner-bockmayr-siebert-15,klarner-streck-siebert-17}
is an elaboration of ``forcing'' mappings, introduced by Kauffman~\cite{kauffman-71,kauffman-72} and further developed by Fogelman-Soulié~\cite{fogelman-soulie-85}.
This algorithm comprises two stages.
The first one
computes the prime implicants~\cite{quine-52,strzemecki-92} of each regulatory function as well as that of the negation of each regulatory function.
The second stage (which may or may not include minimization or maximization) consists in finding subsets of prime implicants that satisfy two conditions: ``consistency'' and ``stability''.
Each such subset corresponds to a trap space.
For the first stage, Klarner uses a \texttt{espresso} and for the second stage Klarner has experimented with two techniques: ASP and integer linear programming; in~\cite{klarner-bockmayr-siebert-15} ASP exhibits a better performance.

\texttt{BioLQM}~\cite{naldi-18} ``proposes an adapted version of the method implemented in \texttt{PyBoolNet} using the clingo ASP solver~\cite{gebser-kaufmann-kaminski-ostrowski-schaub-schneider-11}, and introduces a new alternative implementation based on decision diagrams''.

\texttt{trappist}~\cite{trinh-benhamou-soliman-23,trinh-hiraishi-benhamou-22} is based on siphons in Petri nets.
Intuitively, a siphon is a set of places that once unmarked remains so.

Both \texttt{mpbn}~\cite{chatain-haar-pauleve-18,pauleve-kolcak-chatain-haar-20,trinh-benhamou-pauleve-24} and \texttt{tsconj}~\cite{trinh-benhamou-pastva-soliman-24} employ ASP.
They can be viewed as based on a cover of Boolean networks that uses disjunctive normal forms of both regulatory functions and their negations.

This cover is as follows. 
With each argument $a\in A$, we associate two Boolean variables $p_a$ (for positive) and $n_a$ (for negative).
We then map an expression in disjunctive normal form to an expression over two-valued variables with:
\begin{align*}
  \sigma(0) = & \ 0 \\
  \sigma(1) = & \ 1 \\
  \sigma(a) = & \ p_a\\
  \sigma(\neg a) = & \ n_a\\
  \sigma(\phi\land\psi) = & \ \sigma(\phi)\land\sigma(\psi)\\
  \sigma(\phi\lor\psi) = & \ \sigma(\phi)\lor\sigma(\psi)
\end{align*}
We do likewise to three-valued interpretations defining the map $\sigma' : \set{0,1,u}^A \ra \set{0,1}^{A'}$, where $A'=\set{p_a\mid a\in A}\cup\set{n_a\mid a\in A}$, as follows:
\begin{align*}
  \sigma'(v)(0) = 0 &\\
  \sigma'(v)(1) = 1 &\\
  \sigma'(v)(p_a) = &\ 
    \left\{ \begin{array}{l}
              0 \R{ if } v(a)=0,\\
              1 \R{ if } v(a)\in\set{1,u}
            \end{array}
    \right.\\
  \sigma'(v)(n_a) = &\ 
    \left\{ \begin{array}{l}
              0 \R{ if } v(a)=1,\\
              1 \R{ if } v(a)\in\set{0,u}
            \end{array}
    \right.\\
  \sigma'(\phi\land\psi) = & \ \sigma'(\phi)\land\sigma'(\psi)\\
  \sigma'(\phi\lor\psi) = & \ \sigma'(\phi)\lor\sigma'(\psi)
\end{align*}
Finally, we need a consensus function: % $\tau : \set{(0,1),(1,0),(1,1)} \ra \set{0,1,u}$:
\begin{align*}
  \tau(0,1) = & \ 0\\
  \tau(1,0) = & \ 1\\
  \tau(1,1) = & \ u
\end{align*}
With a slight abuse of notation, we will use $\R{DNF}(\phi)$ for any disjunctive normal form of $\phi$ where no disjunct subsumes another and no disjunct has both a positive and a negative literal of the same variable.
It can be proved that the following equality holds:
\begin{equation*}
\label{eq-tsconj}
  \tau(\ \sigma'(v)\sigma(\R{DNF}(\phi_a)),
       \ \sigma'(v)\sigma(\R{DNF}(\neg\phi_a))\ )
  = \Gamma_\mc{D}(v)(a)
\end{equation*}
which allows us to circumvent the computation of $\phi_a$ on all completions $\sem{v}$ of $v$.

\begin{example} Given the pair $(A,R)$ of Example~\ref{ex-running}, associate with the argument $c\in A$ the propositional formula 
\[
  \phi_c = \neg b\land\neg d. 
\]
Take a three-valued interpretation $v\in\set{0,1,u}^A$ such that $v(b)= 1$ and $v(d) = u$.
Then there exist  $w_1,w_2\in \sem{v}$ such that, for instance, $w_1(b) = 1$ and $w_1(d)= 0$, and $w_2(b) =1$ and $w_2(d) = 1$.
Now, $\Gamma(w_1)(\phi_c) = \Gamma(w_2)(\phi_c) = 0$, so that $\Gamma(v)(\phi_c)=0$.

On the other hand,
\begin{align*}
\R{DNF}(\phi_c) & = \neg b\land\neg d\\
\R{DNF}(\neg\phi_c) & = (\neg b\land d)\lor b
\end{align*}                        
\[
\sigma'(v)(p_b) = 1, \quad \sigma'(v)(n_b) = 0,\quad \sigma'(v)(p_d)= 1,\quad \sigma'(v) (n_d) = 1
\]
\begin{align*}
\sigma(\R{DNF}(\phi_c)) & = \sigma(\neg b\land\neg d) = n_b\land n_d\\
\sigma(\R{DNF}(\neg \phi_c)) & = \sigma((\neg b\land d)\lor b) = (n_b\land p_d)\lor p_b
\end{align*}
\begin{align*}
\sigma'(v)\sigma(\R{DNF}(\phi_c)) & = 0,\\
\sigma'(v)\sigma(\R{DNF}(\neg\phi_c)) & = 1.
\end{align*}
Hence,
\[ \tau(\ \sigma'(v)\sigma(\R{DNF}(\phi_c)),\ \sigma'(v)\sigma(\R{DNF}(\neg\phi_c))\ ) = 0.
\]
\end{example}

\texttt{mpbn} represents what we and \texttt{tsconj} have named $p_a$ and $n_a$ with pairs $(a,\texttt{1})$ and $(a,\texttt{-1})$, respectively.
Neither negation or disjunction occur in the encondings of \texttt{mpbn}.
For signed networks (to be treated in Sect.~\ref{sect-signed}) \texttt{mpbn} does use DNF.
(The fact that a network is signed allows the disposal of the DNF of the negation of the regulatory functions, as we will see in Sect.~\ref{sect-signed}.)
However, for arbitrary Boolean networks it uses BDDs instead.
There is a close connection, nevertheless, between BDDs and DNFs: a path from the root to the leaf labeled with $1$ in a BDD corresponds to a conjunction of literals occurring in a DNF of the function.
Similarly, a path from the root to the leaf labeled with $0$ corresponds to a conjunction of literals occurring in a DNF of the negation of the function.
Unlike \texttt{tsconj}, \texttt{mpbn} does not perform any optimization on the representation of the regulatory functions.

\texttt{tsconj} does not use negation either, but it
does employ disjunction on the left-hand side of rules:
\[
  p_a\lor n_a\la
\]
to exclude the possibility of both $p_a$ and $n_a$ being false.
In addition, for each regulatory function $\phi_a$, the encoding represents the formula $a\leftrightarrow\phi_a$ by first rewriting it as two rules:
\begin{align*}
\sigma(a) & \la \sigma(\R{DNF}(\phi_a))\\
\sigma(\neg a) & \la \sigma(\R{DNF}(\neg\phi_a))
\end{align*}
splitting the biimplication into the implication to the left and the contrapositive of the implication to the right.
Next, each disjunct is translated into a different rule by adding auxiliary variables.

ASP solvers compute \emph{minimal} solutions, in the sense of set containment.
Hence, with this encoding, the ASP solver computes minimal trap spaces.

\texttt{tsconj} has a number of optimizations to this encoding.
First, there are cases where $\phi$ can be treated verbatim, without having to compute a DNF (when a function satisfies a sufficient condition called ``safeness'').
Next, there may be repeated subformulas in the encoding that can be factorized and some auxiliary variables that can be bypassed.

\paragraph{Attractors.}
In gene regulation, we are, in general, ultimately interested not in minimal trap spaces per se, but in attractors instead.
\texttt{PyBoolNet}~\cite{klarner-siebert-15} uses model checking and \texttt{mtsNFVS}~\cite{trinh-hiraishi-benhamou-22} first reduces the state-transition graph, and then uses reachability analysis.
We refer the reader to~\cite{trinh-hiraishi-benhamou-22} and references therein for methods for finding attractors from minimal trap spaces.

\subsection{Benchmarks}
We tested the performance of \texttt{k++ADF} and \texttt{tsconj} on a number of ADFs/BNs.

We used:
\items
{three databases from \texttt{YADF}\footnote{\url{https://www.dbai.tuwien.ac.at/proj/adf/yadf/}}:
  \texttt{ABA2AF}, \texttt{Planning2AF}, and \texttt{Traffic}, with 100 instances each database and with between 10 and 300 arguments each instance~\cite{linsbichler-maratea-niskanen-wallner-woltran-18}; and
}
{two databases used also by \texttt{tsconj}\footnote{\url{https://zenodo.org/records/10406324}}:
(a) the Biodivine Boolean Models (BBM) dataset~\cite{pastva-safranek-benes-brim-henzinger-23}, with 212 real-world Boolean networks with up to 321 variables and
(b) 39 real-world models, selected by the authors of~\cite{trinh-benhamou-pastva-soliman-24}, which do not appear in BBM, with up to 3,158 variables.
}

We only registered instances for which both tools were able to compute all preferred interpretations/minimal trap spaces with the allotted resources.
The reason is that we wished to compare that both tools computed the same answers for each instance (which was the case).
This termination criterion was only met by 541 instances; in particular, the some instances did not (i.e., only one of the two tools terminated within 64 Gb of RAM and 64 minutes).
Figure~\ref{fig-k++adf-vs-tsconj} shows the results.\footnote{We used a core belonging to a cluster with 13 nodes on two processors Intel\textregistered\ Xeon\textregistered\ Gold 5218R CPU @ 2.10 GHz (40 cores per node) and 256 Gb of RAM, seven nodes with two processors Intel Xeon CPU E5-2690 v3 @ 2.60 GHz (24 cores per node) and 128 Gb of RAM.}

\begin{figure}[ht]
\begin{center}
    \includegraphics[scale=0.8]{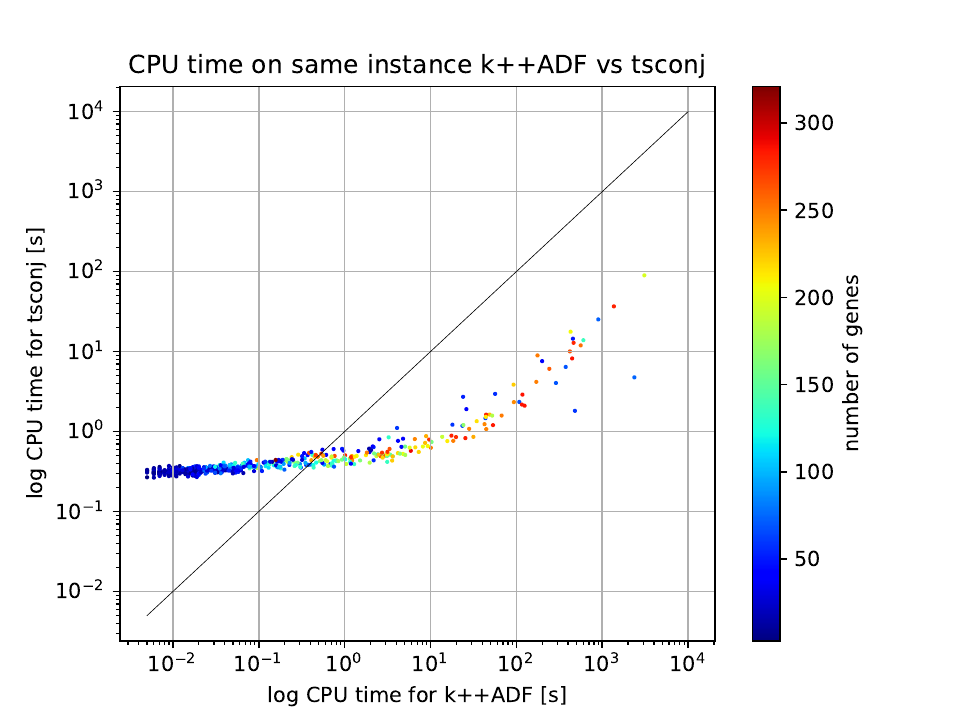}
\end{center}
\caption{\label{fig-k++adf-vs-tsconj} \texttt{k++ADF} vs \texttt{tsconj}.}
\end{figure}

\section{Signed reasoning}
\label{sect-signed}

\subsection{Bipolarity in argumentation}
The additional expressive power of ADFs compared with Dung frameworks comes at a price (ADFs are one level up in the ``polynomial hierarchy''~\cite{papadimitriou-94} than Dung frameworks).
However, ``bipolar'' ADFs (or BADFs), a restriction of ADFs, have the same computational complexity for many problems as Dung frameworks, while being strictly more expressive~\cite{baumann-heinrich-23,strass-wallner-15}.

(When developing the stable, admissible, and preferred semantics for ADFs, Brewka and Woltran needed~\cite{brewka-woltran-10} an operator to have a least fixed point.
Hence, these authors made it a requirement for the acceptance conditions to be monotone,
called ``supporting'' acceptance conditions;
anti-monotone acceptance conditions, called ``attacking'' acceptance conditions,
were also incorporated without affecting the semantics.
Hence, in~\cite{brewka-woltran-10} these semantics are only defined for BADFs;
subsequent work~\cite{brewka-ellmauthaler-strass-wallner-woltran-13,strass-13} defined these notions for arbitrary ADFs.)

Similarly constrained frameworks have been independently developed within abstract argumentation~\cite{amgoud-cayrol-lagasquie-schiex-08,cayrol-lagasquie-schiex-13,cohen-gottifredi-garcia-simari-14,nouioua-risch-10,polberg-oren-14}.

For each two-valued function $x\in\set{0,1}^A$, each argument $a\in A$ and each value $p\in\set{0,1}$, we define the two-valued function $x^{a\mt p} \in \set{0,1}^A$ given by  $x^{a\mt p}(a) = p$ and $x^{a\mt p}(b)=x(b)$ for all $b\neq a$.
An acceptance condition $\phi_b$ is \emph{monotone} on argument $a$, or the ordered pair $(a,b)$ is \emph{supporting}, which we write as $aM^+b$, if $\phi_b$ depends on $a$ and for each $p,q\in\set{0,1}$,
\begin{equation}
\label{eq-monot}
 p\leq q \quad\R{ implies }\quad  x^{a\mt p}(\phi_b) \leq x^{a\mt q}(\phi_b) \, \R{ for all }\, x\in \set{0,1}^A.
\end{equation}
An acceptance condition $\phi_b$ is \emph{anti-monotone} on argument $a$, or the ordered pair $(a,b)$ is \emph{attacking}, which we write as $aM^-b$, if $\phi_b$ depends on $a$ and for each $p,q\in\set{0,1}$,
\begin{equation}
\label{eq-anti-monot}
  p\leq q \quad\R{ implies }\quad  x^{a\mt p}(\phi_b) \geq x^{a\mt q}(\phi_b) \, \R{ for all }\, x\in \set{0,1}^A.
\end{equation}
The \emph{attacking} relation $M^-$ is a generalization of Dung's \emph{attack} relation $R$.
To see this, observe first that we can view an arbitrary acceptance condition $\phi_b$ (which is a function) as having four contingencies:
\begin{align}
x^{a\mt 0}(\phi_b)=0 &\ \R{ and }\ x^{a\mt 1}(\phi_b)=0 \label{cont-00} \\
x^{a\mt 0}(\phi_b)=0 &\ \R{ and }\ x^{a\mt 1}(\phi_b)=1 \label{cont-01} \\
x^{a\mt 0}(\phi_b)=1 &\ \R{ and }\ x^{a\mt 1}(\phi_b)=0 \label{cont-10} \\
x^{a\mt 0}(\phi_b)=1 &\ \R{ and }\ x^{a\mt 1}(\phi_b)=1 \label{cont-11}
\end{align}
Next, note that only contingency \rf{cont-01} is excluded in \emph{attacking}.
Now, recall that in Dung's \emph{attack} relation (where the acceptance conditions are the conjunction of the negation of the parents, i.e., $x(\phi_b)=1$ only when $a=0$ for all $a\in\R{par}(b)$), there are only two contingencies:
\rf{cont-00} and \rf{cont-10},
both of which occur in \emph{attacking}.
We conclude that \emph{attacking} generalizes \emph{attack}.

\subsection{Signed interaction in Boolean networks}
Independently and before the advent of abstract argumentation, signed reasoning has been used in gene regulation modeling~\cite{kauffman-71,kauffman-74,thomas-d-ari-90}.
Unlike in argumentation, we find two notions of signed reasoning in the literature of Boolean gene regulation: one based on monotonicity, as in argumentation, but another one based on the discretization of the time derivative.
We will see that these two notions are equivalent to each other,
provided that we rule out ``dual'' derivative regulations, that is, provided that we take strict derivative-based interactions (it is not necessary to take strict monotonicity-based interactions).

\paragraph{Monotonicity-based interaction.}
The monotonicity-based notion dates back to Waddington's canalization~\cite{waddington-42}.
As an abstraction of the Michaelis--Menten equation~\cite{deichmann-schuster-mazat-cornish-bowden-14,michaelis-menten-13}, which is monotonic, Kauffman reinterpreted~\cite{kauffman-71,kauffman-74}, for Boolean functions, Waddington's concept (initially called ``forcible'') as follows~\cite{li-adeyeye-murrugarra-aguilar-laubenbacher-13}:
The acceptance condition $\phi_a$ is \emph{canalizing} on $b\in A$ if $x^{b\mt p}(\phi_a)=q$ for $p,q\in\set{0,1}$.
There are four cases in this definition (two for $p$ and two for $q$), which can easily be seen equivalent to monotonicity (two cases) and anti-monotonicity (two cases).
Experimental evidence confirms that in Nature there is a strong bias towards monotonic regulation~\cite{harris-sawhill-wuensche-kauffman-02,subbaroyan-margin-samal-22}.
Literature examples of monotonic signed regulation are:~\cite{akutsu-melkman-tamura-yamamoto-11,aracena-08,aracena-richard-salinas-17,benevs-brim-kadlecaj-pastva-safranek-20,li-adeyeye-murrugarra-aguilar-laubenbacher-13,moon-lee-pauleve-23,mori-mochizuki-17,veliz-cuba-11}.

A monotone (resp.\ anti-monotone) function can also be characterized as one in which only positive (resp.\ negative) literals occur~\cite{balogh-dong-lidicky-mani-zhao-23}.
A simplification in \texttt{mpbn} for monotone or anti-monotone regulatory functions is that~\rf{eq-tsconj} reduces to either the positive case or the negative case of $\phi_a$.

\paragraph{Derivative-based interaction.}
Instead of monotonicity-based interactions, many authors use a \emph{derivative-based} notion of signed regulation.
The ubiquitous use of differential equations as models of gene regulation justifies this approach.
The Boolean derivative $\R{D}_a\phi_b \colon \set{0,1}^{\R{par}(b)}\ra\set{-1,0,1}$~\cite[p 47]{o-donell-14} of $\phi_b$ with respect to $a\in \R{par}(b)$ is:
\[
  \R{D}_ax(\phi_b) = x^{a\mt 1}(\phi_b) - x^{a\mt 0}(\phi_b)
\]
which is normally existentially quantified, resulting in a signed interaction graph~\cite{aracena-cabrera-crot-salinas-21,comet-et-al-13,munoz-carrillo-azpeitia-rosenblueth-18,naldi-thieffry-chaouiya-07,richard-19,richard-rossignol-comet-bernot-guespin-michel-merieau-12,richard-ruet-13} as follows.

Given an acceptance condition $\phi_b$, the ordered pair $(a,b)$ is an \emph{activation}, which we write as $aD^+b$, if $\phi_b$ depends on $a$ and for each $p,q\in\set{0,1}$,
\begin{equation}
\label{eq-pos-deriv}
   p< q \quad\R{and}\quad  x^{a\mt p}(\phi_b) < x^{a\mt q}(\phi_b) \, \R{ for some }\, x\in \set{0,1}^A.
\end{equation}
Given an acceptance condition $\phi_b$, the ordered pair $(a,b)$ is an \emph{inhibition}, which we write as $aD^-b$, if $\phi_b$ depends on $a$ and for each $p,q\in\set{0,1}$,
\begin{equation}
\label{eq-neg-deriv}
     p< q \quad\R{and}\quad  x^{a\mt p}(\phi_b) > x^{a\mt q}(\phi_b) \, \R{ for some }\, x\in \set{0,1}^A.
\end{equation}

\paragraph{Relationship between monotonicity-based and derivative-based interactions.}
First, observe that monotonicity-based interaction implies derivative-based interaction:
\begin{align*}
  aM^-b &\ \R{ implies }\ aD^-b,\\
  aM^+b &\ \R{ implies }\ aD^+b .
\end{align*}
To see this, recall that $M^-$ excludes contingency~\rf{cont-01}.
However, by definition of Boolean network, the acceptance condition of an argument depends on all the parents of that argument.
Hence, contingency~\rf{cont-10} must occur in $M^-$, which is the existentially quantified contingency of $D^-$.
Consequently, $aM^-b \R{ implies } aD^-b$.

Next, note that the positive derivative $D^+$ is  equivalent to the negation of anti-monotonicity $M^-$:
\[
aD^+b \quad \R{ if and only if }\quad \R{ not }\, aM^-b.
\] 
To see this, assume that $\phi_b$ depends on $a$, and observe that  %it is not the case that $aM^-b$ 
$\phi_b$ is not anti-monotone on $a$ (i.e., not $aM^-b$)
if and only if   $p\leq q$  and $x^{a\mt p}(\phi_b) < x^{a\mt q}(\phi_b)$  for some $x\in \set{0,1}^A$. 
Since $\phi_b$ depends on $a$, we have that $x^{a\mt p}(\phi_b) \neq x^{a\mt q}(\phi_b)$ if and only if $p\neq q$. 
Hence,  not $aM^-b$ is equivalent to  $p< q$  and $x^{a\mt p}(\phi_b) < x^{a\mt q}(\phi_b)$  for some $x\in \set{0,1}^A$, which in turn is equivalent to $aD^+b$. 
Similarly, the negative derivative $D^-$ is  equivalent to the negation of monotonicity $M^+$:
\[
aD^-b \quad \R{ if and only if }\quad \R{ not }\, aM^+b.
\]

The fact that we can transit from monotonicity to derivation through negation suggests considering \emph{strict} interactions (used, for instance, in~\cite{baumann-strass-17}) so that strict monotonicity is  equivalent to the strictly positive derivative.
We thus get:
\begin{equation}
  aM^+b\, \R{ and  not }\, aM^-b \quad \R{ if and only if }\quad \R{ not }\, aD^-b \, \R{ and }\, aD^+b \label{eq-strict}
\end{equation}
which reconciles monotonicity-based and derivation-based interactions: strict monotonicity is  equivalent to the strict positive derivative.
Observe, nevertheless, that by definition of Boolean network, the acceptance condition of an argument depends on all the parents of that argument.
Hence, contingency~\rf{cont-01}, which is existentially quantified in $D^+ ( = \neg M^-)$, does occur in $M^+$.
Therefore, to obtain strict monotonicity $M^+$ we do not have to take the conjunction with the negation of anti-monotonicity $\neg M^-$.
By contrast, to obtain the strict positive derivative $D^+$ we do have to take the conjunction with the negation of the negative derivative $\neg D^-$.

\paragraph{Conjunction of positive and negative interactions.}
Having combined positive and negative interactions, as in~\rf{eq-strict}, invites exploring other Boolean combinations.
In the case of monotonicity, an interaction satisfying $M^+\land M^-$ is %equivalent to
``irrelevant'',
consisting only of contingencies~\rf{cont-00} and~\rf{cont-11}.
Such interactions are not interesting, as they are excluded by definition.
Perhaps that is a reason why other Boolean combinations of interactions have been neglected.

Taking the conjunction of positive and negative derivatives is different.
In an interaction satisfying $aD^-b$ and $aD^+b$, often referred to as ``ambiguous'', ``paradoxical'', or ``dual'', both contingencies~\rf{cont-01} and~\rf{cont-10} do occur.
It is customary to eliminate such interactions in gene-network modeling, as they rarely appear in Nature.
This %natural
phenomenon justifies taking strict negative or positive interactions, as in~\rf{eq-strict},
thus obtaining a coincidence between both notions of interaction.

For monotonicity-based regulation, ``non-monotonic'' would be a preferable term, since such an interaction is neither monotonic nor anti-monotonic.

\paragraph{Single interactions.}
Further differences between both kinds of signed interaction appear when we take an interaction by itself.
Note first that an interaction satisfying $M^+$ does accept the possibility of the absence of interaction as a result of the definitions of monotone~\rf{eq-monot} and antimonotone~\rf{eq-anti-monot} interactions accepting equalities.
Hence, $M^+$ denotes a possibly \emph{irrelevant} interaction (i.e., possibly satisfying $M^-\land M^+$).
We can name this a ``strict, optional, and positive'' interaction.

By contrast, an interaction satisfying $D^+$
accepts the possibility of a \emph{dual} interaction (i.e., possibly satisfying $D^-\land D^+$).
We can term this a ``possibly dual, mandatory, and positive'' interaction.
\emph{This suggests enriching the labels of an ordinary interaction graph with Boolean combinations of either $M^+$ and $M^-$ or $D^+$ and $D^-$.}

\paragraph{Interaction graph redux.}
Tables~\ref{tab-griffin-mon} and~\ref{tab-griffin-der} respectively list all non-trivial possibilities for monotonicity-based and derivative-based interactions.
Figure~\ref{fig-venn-4} summarizes these tables. %\ref{fig-squares}: each line of the tables corresponds to a subset of squares.

Following Nature, we would discard non-monotonic/dual and retain strict regulations;
that is, we could use six possibilities of Tables~\ref{tab-griffin-mon} and~\ref{tab-griffin-der}.
In gene regulation, where not all experimental information may be available, being able to specify such families of networks is both useful and important.
For example, it may be suspected that a certain interaction exists, without having hard evidence; in such a case, an optional interaction may be used.
It may also happen that we know that a certain interaction between two genes exists, but there is no indication as to the sign of such an interaction; a formula with ``unknown sign'' can then be employed.
A tool based on a SAT solver and using Boolean combinations of $D^-$ and $D^+$ (where the enriched interaction graphs are called R-graphs) is described in~\cite{munoz-carrillo-azpeitia-rosenblueth-18}.

The different meanings of both notions of regulation attest to the fact that, although these notions can be made to coincide, important differences appear when Boolean combinations are considered.

\begin{table}
\begin{center}
\begin{tabular}{ll}
\multicolumn{1}{c}{\emph{label monotonicity}} & \emph{intuition} \\ \hline\hline\hline
$M^+\land\neg M^-$ & strict, mandatory, positive \\
$M^-\land\neg M^+$ & strict, mandatory, negative\\
$M^-\oplus M^+$ & strict, mandatory, unknown sign\\
$M^+$ & strict, optional, positive\\
$M^-$ & strict, optional, negative\\
$M^-\lor M^+$ & strict, optional, unknown sign\\
$\neg M^-\land\neg M^+$ & non-mon., mandatory\\
$M^-\leftrightarrow M^+$ & non-mon., optional\\
$\neg M^-$ & poss.\ non-mon., mandatory, positive\\
$\neg M^+$ & poss.\ non-mon., mandatory, negative\\
$\neg(M^-\land M^+)$ & poss.\ non-mon., mandatory, unknown sign\\
$\neg M^-\lor M^+$ & poss.\ non-mon., optional, positive\\
$M^-\lor\neg M^+$ & poss.\ non-mon., optional, negative\\
$M^-\land M^+$ & irrelevant\\
$\textit{true}$ & poss.\ non-mon., optional, unknown sign
\end{tabular}
\end{center}
\caption{\label{tab-griffin-mon} Non-trivial Boolean combinations of monotonicity-based interactions.} %Kopka & Daly, p 178
\end{table}

\begin{table}
\begin{center}
\begin{tabular}{ll}
\multicolumn{1}{c}{\emph{label derivative}} & \emph{intuition} \\ \hline\hline\hline
$D^+\land\neg D^-$ & strict, mandatory, positive\\
$D^-\land \neg D^+$ & strict{,} mandatory{,} negative\\
$D^-\oplus D^+$ & strict{,} mandatory{,} unknown sign\\
$D^+$ & poss.\ dual{,} mandatory{,} positive\\
$D^-$ & poss.\ dual{,} mandatory{,} negative\\
$D^-\lor D^+$ & poss.\ dual{,} mandatory{,} unknown sign\\
$\neg D^-\land \neg D^+$ & irrelevant\\
$D^+\leftrightarrow D^-$ & dual{,} optional\\
$\neg D^-$ & strict{,} optional{,} positive\\
$\neg D^+$ & strict{,} optional{,} negative\\
$\neg(D^+\land D^-)$ & strict{,} optional{,} unknown sign\\
$D^+\lor\neg D^-$ & poss.\ dual{,} optional{,} positive\\
$D^-\lor\neg D^+$ & poss.\ dual{,} optional{,} negative\\
$D^+\land D^-$ & dual{,} mandatory\\
$\textit{true}$ & poss.\ dual, optional{,} unknown sign
\end{tabular}
\end{center}
\caption{\label{tab-griffin-der} Non-trivial Boolean combinations of derivative-based interactions~\cite{munoz-carrillo-azpeitia-rosenblueth-18}.} %Kopka & Daly, p 178
\end{table}

\begin{figure}
\begin{center}
\includegraphics[scale=0.5]{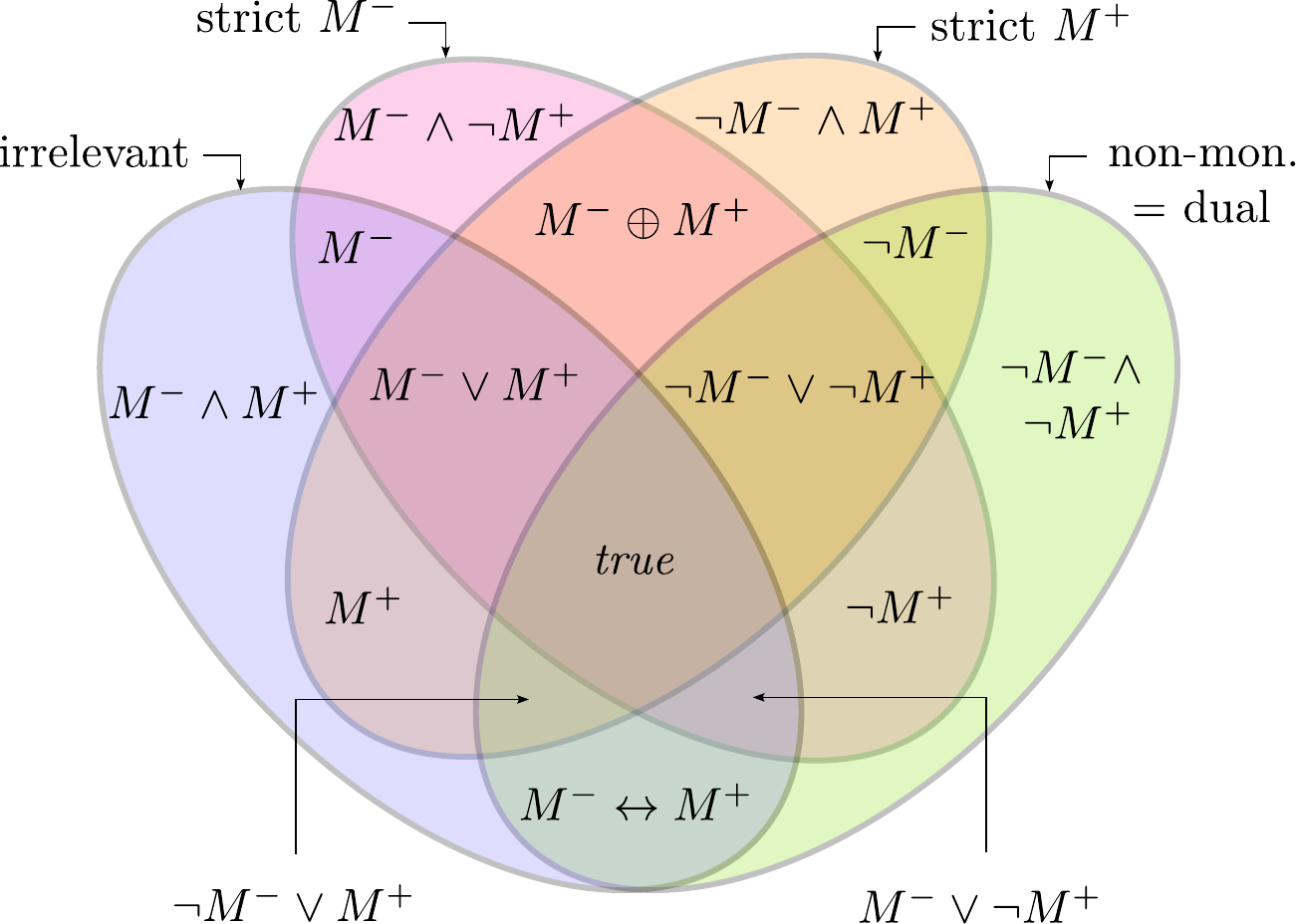}
\end{center}
\caption{\label{fig-venn-4} Diagram of edge labels for monotonicity-based interaction. The corresponding diagram for derivative-based interaction is obtained by substituting $\neg D^-$ for $M^+$ and $\neg D^+$ for $M^-$.}
\end{figure}

\section{Concluding remarks}
\label{sect-concl}
We started by establishing a firm basis of ADFs in Dung's formalism through Caminada's approach.
Next, we gave the same definition for both ADF and BN, thus bypassing having to state the most basic coincidences between these formalisms.
We then explained away the notion of conflict freedom.
In addition, we covered several algorithms for computing preferred interpretations and minimal trap spaces, and reported the performance of two tools.
Moreover, inspired by the fact that in the BN-network literature two signed notions occur (namely monotonicity based and derivative-based), we studied the differences between such notions and advocated an enriched influence graph, enabling reasoning about incomplete frameworks/networks.

\paragraph{Related work.}
Independently from us, Heyninck, Knorr, and Leite have also observed~\cite{heyninck-knorr-leite-24} connections between ADFs and BNs.
After noting similarities between both formalisms, these authors establish the coincidence between (a) three-valued interpretations and symbolic states, (b) admissible extensions and trap spaces, and (c) preferred extensions and minimal trap spaces, along with some differences, such as the lack of the notion of ``self-supporting'' arguments in Boolean networks and the absence in argumentation of the state-transition graph.
Subsequently, they remark the common occurrence of derivative-based reasoning and conclude by asserting that complexity results in one area also hold in the other area.

\subsection{Coincidence of preferred interpretations with minimal trap spaces}
Ideally, Dung abstract argumentation is interested in stable extensions, i.e., conflict-free sets of arguments that attack all other arguments.
As this notion is often too stringent, it is normally relaxed into that of sets of arguments that attack all arguments that attack the set, i.e., admissible sets.
When the attack relation is generalized to arbitrary Boolean functions, ADFs, this notion becomes that of admissible interpretations.
In this concept, there is no notion of time.

Boolean networks, on the other hand, are an offshoot of dynamical systems, where time is pivotal.
Often a Boolean network is defined as a function $f : \set{0,1}^n \ra \set{0,1}^n$, and an application of $f$ is seen as the lapse of one time step.
Time, in turn, is tied to the notions of transitions, update disciplines, and attractors.
As an approximation to attractors, trap spaces were proposed.
A trap space is a three-valued interpretation $v$ such that for each of its completions $x\in\sem{v}$, $f(x)$ is also one of its completions.
But trap spaces are independent of the update discipline.
Thus the coincidence of (a) trap spaces and admissible interpretations and of (b) minimal trap spaces and maximal admissible (i.e., preferred) interpretations.

\subsection{Further directions for research}
In spite of sharing many concerns, abstract dialectical argumentation and Boolean gene networks have interacted little with each other as research fields.
Hence, it is not surprising that some concerns be more developed in one area than in the other, which creates avenues for research.

\subsubsection{Parallel developments}

\paragraph{Computation.}
In our benchmarks, the Boolean-network tool we selected, \texttt{tsconj}, performed better for larger networks than the chosen ADF tool, \texttt{k++ADF}, for preferred interpretations/minimal trap spaces.
There are, however, other semantics of interest in argumentation that have not arisen in the field of Boolean gene networks, such as semi-stable, naive, and stage semantics.
Hence, the question of how to extend \texttt{tsconj} to compute these other semantics emerges.

Baumann and Heinrich have observed that, for BADFs, it is possible to bypass the direct implementation of the definition of the characteristic operator $\Gamma_\mc{D}$ with Kleene's strong three-valued logic~\cite{baumann-heinrich-23,kleene-52}.
Such a contribution might impact
tools for Boolean networks.

\paragraph{Repair.}
In general, when new observations show a discrepancy between reality and our model of reality, we must repair the model.
Boolean gene networks are normally built gradually as experiments are performed, hence the model must be repaired.
Repair techniques are also of interest in argumentation~\cite{baumann-ulbricht-19}, where arguments may become available incrementally.

\subsubsection{Argumentation benefitted from BNs}

\paragraph{Abstraction.}
Abstraction formalized by Galois connections, as existing in gene regulation~\cite{fages-soliman-08}, is yet to be studied in abstract argumentation.

\paragraph{Analysis based on interaction graphs.}
Static analysis based on interaction graphs~\cite{aracena-08,mori-mochizuki-17,pauleve-richard-12,remy-ruet-08,remy-ruet-thieffry-08,richard-10,richard-19,richard-comet-07} is possibly more developed in Boolean networks.
Inspired by Thomas's conjectures~\cite{thomas-d-ari-90},
%on necessary conditions for the existence of 
%necessary conditions and sufficient conditions for the existence of fixed points of the 
this line of research attempts to obtain properties of the state transition graph of a BN from its interaction graph.
%Although most contributions apply to asynchronous Boolean networks, and would therefore not necessarily be relevant to argumentation, some results do include synchronous networks and hence are possibly valid for ADFs.
Hence, another direction of research would be to determine which of these results transfer to ADFs.

\paragraph{Boolean network decomposition.}
Decomposition algorithms, aimed at being able to handle larger systems have been developed~\cite{su-pang-paul-19} for Boolean networks, and would probably also be of interest in argumentation.

\paragraph{Counterfactual reasoning.}
In Boolean-network reasoning, hypothetical situations are often consideren as \emph{mutations} (where one gene is stuck at a particular value).
These situations would correspond to counterfactual reasoning in logic.

\paragraph{Nested canalization.}
The Boolean-network notion of canalization, which occurs in ADFs as monotonicity, has been further developed as ``nested canalization''~\cite{jarrah-raposa-laubenbacher-07,li-adeyeye-murrugarra-aguilar-laubenbacher-13}.
It is natural to wonder, therefore, if such an extended concept might have a meaning in argumentation.

\paragraph{Opinion diffusion.}
The role of time in Boolean networks strongly suggests investigating the extension of ADFs with time.
This would have close connections with opinion diffusion viewed as a generalization of argumentation~\cite{grandi-lorini-perrussel-15}.

\paragraph{Reduction.}
Finally, reduction techniques~\cite{klarner-14,naldi-remi-thieffry-chaouiya-11,saadatpour-albert-reluga-13,veliz-cuba-11} have existed for some time in Boolean networks.
They shrink the network while preserving properties of interest.
These would be valuable tools for argumentation.

\subsubsection{BNs benefitted from argumentation}
\paragraph{Different kinds of support.}
Yet another terrain worth investigating is that of the different kinds of support~\cite{amgoud-cayrol-lagasquie-schiex-08,cohen-gottifredi-garcia-simari-14,nouioua-risch-10,polberg-oren-14} that have emerged in argumentation, suggesting that considering such kinds of support might also have a meaning in gene regulation.

\paragraph{Graded argumentation.}
Graded argumentation~\cite{grossi-modgil-19} is a generalization of Dung's formalism that has not been explored for Boolean networks, as far as we know.

\paragraph{Logic.}
From its advent, abstract argumentation has remained close to logic~\cite{besnard-cayrol-lagasquie-schiex-20,boella-hulstijn-van-der-torre-05,caminada-gabbay-09,grossi-09,prakken-vreeswijk-02,proietti-yuste-ginel-21,schwarzentruber-vesic-rienstra-12}.
These works may inspire logical approaches to Boolean networks.

\section*{Acknowledgments}
We would like to express gratitude to
Florian Bridoux, Miguel Carrillo, Victor David, Christopher Leturc, and Adrien Richard with whom we had fruitful discussions.
We gladly mention Johannes Wallner's help, who directed us to \texttt{k++ADF}.
Finally,
we would like to thank the Centro de Ciencias de la Complejidad of the Universidad Nacional Aut\'onoma de M\'exico for having granted us access to their cluster to perform our benchmarks; Jos\'e Luis Gordillo helped us configure and use the cluster.

\bibliographystyle{plain}
\bibliography{review}

\begin{thebibliography}{100}

\bibitem{abou-jaoude-et-al-16}
Wassim Abou-Jaoud{\'e}, Pauline Traynard, Pedro~T Monteiro, Julio
  Saez-Rodriguez, Tom{\'a}{\v{s}} Helikar, Denis Thieffry, and Claudine
  Chaouiya.
\newblock Logical modeling and dynamical analysis of cellular networks.
\newblock {\em Frontiers in Genetics}, 7:94, 2016.

\bibitem{akman-watterson-parton-binns-millar-ghazal-12}
Ozgur~E Akman, Steven Watterson, Andrew Parton, Nigel Binns, Andrew~J Millar,
  and Peter Ghazal.
\newblock Digital clocks: simple {B}oolean models can quantitatively describe
  circadian systems.
\newblock {\em Journal of The Royal Society Interface}, 9(74):2365--2382, 2012.

\bibitem{akutsu-18}
Tatsuya Akutsu.
\newblock {\em Algorithms for analysis, inference, and control of {B}oolean
  networks}.
\newblock World Scientific, 2018.

\bibitem{akutsu-melkman-tamura-yamamoto-11}
Tatsuya Akutsu, Avraham~A Melkman, Takeyuki Tamura, and Masaki Yamamoto.
\newblock Determining a singleton attractor of a {B}oolean network with nested
  canalyzing functions.
\newblock {\em Journal of Computational Biology}, 18(10):1275--1290, 2011.

\bibitem{albert-thakar-li-zhang-albert-08}
Istv{\'a}n Albert, Juilee Thakar, Song Li, Ranran Zhang, and R{\'e}ka Albert.
\newblock Boolean network simulations for life scientists.
\newblock {\em Source code for biology and medicine}, 3(1):1--8, 2008.

\bibitem{albert-thakar-14}
Reka Albert and Juilee Thakar.
\newblock Boolean modeling: a logic-based dynamic approach for understanding
  signaling and regulatory networks and for making useful predictions.
\newblock {\em Wiley Interdisciplinary Reviews: Systems Biology and Medicine},
  6(5):353--369, 2014.

\bibitem{alon-19}
Uri Alon.
\newblock {\em An Introduction to Systems Biology: Design Principles of
  Biological Circuits}.
\newblock Chapman and Hall/CRC, 2019.

\bibitem{amgoud-cayrol-lagasquie-schiex-08}
Leila Amgoud, Claudette Cayrol, Marie-Christine Lagasquie-Schiex, and Pierre
  Livet.
\newblock On bipolarity in argumentation frameworks.
\newblock {\em International Journal of Intelligent Systems},
  23(10):1062--1093, 2008.

\bibitem{aracena-08}
Julio Aracena.
\newblock Maximum number of fixed points in regulatory {B}oolean networks.
\newblock {\em Bulletin of mathematical biology}, 70:1398--1409, 2008.

\bibitem{aracena-cabrera-crot-salinas-21}
Julio Aracena, Luis Cabrera-Crot, and Lilian Salinas.
\newblock Finding the fixed points of a {B}oolean network from a positive
  feedback vertex set.
\newblock {\em Bioinformatics}, 37(8):1148--1155, 2021.

\bibitem{aracena-richard-salinas-17}
Julio Aracena, Adrien Richard, and Lilian Salinas.
\newblock Number of fixed points and disjoint cycles in monotone {B}oolean
  networks.
\newblock {\em SIAM Journal on Discrete Mathematics}, 31(3):1702--1725, 2017.

\bibitem{balogh-dong-lidicky-mani-zhao-23}
J{\'o}zsef Balogh, Dingding Dong, Bernard Lidick{\`y}, Nitya Mani, and Yufei
  Zhao.
\newblock Nearly all $k$-{SAT} functions are unate.
\newblock In {\em Proceedings of the 55th Annual ACM Symposium on Theory of
  Computing}, pages 958--962, 2023.

\bibitem{baroni-gabbay-giacomin-van-der-torre-18}
Pietro Baroni, Dov Gabbay, Massimilino Giacomin, and Leendert Van~der Torre,
  editors.
\newblock {\em Handbook of formal argumentation}.
\newblock College Publications, 2018.

\bibitem{baumann-heinrich-20}
Ringo Baumann and Maximilian Heinrich.
\newblock Timed abstract dialectical frameworks: {A} simple translation-based
  approach.
\newblock In {\em Computational Models of Argument: Proceedings of COMMA 2020},
  volume 326, pages 103--110. IOS Press, 2020.

\bibitem{baumann-heinrich-23}
Ringo Baumann and Maximilian Heinrich.
\newblock Bipolar abstract dialectical frameworks are covered by {K}leene’s
  three-valued logic.
\newblock In {\em Thirty-Second International Joint Conference on Artificial
  Intelligence}, volume~3, pages 3123--3131, 2023.

\bibitem{baumann-strass-17}
Ringo Baumann and Hannes Strass.
\newblock On the number of bipolar {B}oolean functions.
\newblock {\em Journal of Logic and Computation}, 27(8):2431--2449, 2017.

\bibitem{baumann-ulbricht-19}
Ringo Baumann and Markus Ulbricht.
\newblock If nothing is accepted---{R}epairing argumentation frameworks.
\newblock {\em Journal of Artificial Intelligence Research}, 66:1099--1145,
  2019.

\bibitem{bench-capon-dunne-06}
Trevor~JM Bench-Capon and Paul~E Dunne.
\newblock Argumentation in {AI} and law: Editors' introduction.
\newblock {\em Artificial Intelligence and Law}, 13:1--8, 2006.

\bibitem{bench-capon-dunne-07}
Trevor~JM Bench-Capon and Paul~E Dunne.
\newblock Argumentation in artificial intelligence.
\newblock {\em Artificial intelligence}, 171(10-15):619--641, 2007.

\bibitem{benevs-brim-kadlecaj-pastva-safranek-20}
Nikola Bene{\v{s}}, Lubo{\v{s}} Brim, Jakub Kadlecaj, Samuel Pastva, and David
  {\v{S}}afr{\'a}nek.
\newblock {AEON}: attractor bifurcation analysis of parametrised {B}oolean
  networks.
\newblock In {\em Computer Aided Verification: 32nd International Conference,
  CAV 2020, Los Angeles, CA, USA, July 21--24, 2020, Proceedings, Part I 32},
  pages 569--581. Springer, 2020.

\bibitem{besnard-cayrol-lagasquie-schiex-20}
Philippe Besnard, Claudette Cayrol, and Marie-Christine Lagasquie-Schiex.
\newblock Logical theories and abstract argumentation: {A} survey of existing
  works.
\newblock {\em Argument \& Computation}, 11(1-2):41--102, 2020.

\bibitem{bezem-grabmayer-walicki-12}
Marc Bezem, Clemens Grabmayer, and Micha{\l} Walicki.
\newblock Expressive power of digraph solvability.
\newblock {\em Annals of Pure and Applied Logic}, 163(3):200--213, 2012.

\bibitem{bhattacharyya-et-al-21}
Arnab Bhattacharyya, Ashutosh Gupta, Lakshmanan Kuppusamy, Somya Mani, Ankit
  Shukla, Mandayam Srivas, and Mukund Thattai.
\newblock A formal methods approach to predicting new features of the
  eukaryotic vesicle traffic system.
\newblock {\em Acta Informatica}, 58:57--93, 2021.

\bibitem{biere-heule-van-maaren-walsh-09}
Armin Biere, Marijn Heule, Hans van Maaren, and Toby Walsh, editors.
\newblock {\em Handbook of Satisfiability}.
\newblock IOS Press, 2009.

\bibitem{boella-hulstijn-van-der-torre-05}
Guido Boella, Joris Hulstijn, and Leendert Van Der~Torre.
\newblock A logic of abstract argumentation.
\newblock In {\em International Workshop on Argumentation in Multi-Agent
  Systems}, pages 29--41. Springer, 2005.

\bibitem{bondarenko-dung-kowalski-toni-97}
Andrei Bondarenko, Phan~Minh Dung, Robert~A Kowalski, and Francesca Toni.
\newblock An abstract, argumentation-theoretic approach to default reasoning.
\newblock {\em Artificial intelligence}, 93(1-2):63--101, 1997.

\bibitem{brewka-diller-heissenberger-linsbichler-woltran-20}
Gerhard Brewka, Martin Diller, Georg Heissenberger, Thomas Linsbichler, and
  Stefan Woltran.
\newblock Solving advanced argumentation problems with answer set programming.
\newblock {\em Theory and Practice of Logic Programming}, 20(3):391--431, 2020.

\bibitem{brewka-eiter-truszczynski-11}
Gerhard Brewka, Thomas Eiter, and Miros{\l}aw Truszczy{\'n}ski.
\newblock Answer set programming at a glance.
\newblock {\em Communications of the ACM}, 54(12):92--103, 2011.

\bibitem{brewka-ellmauthaler-strass-wallner-woltran-13}
Gerhard Brewka, Stefan Ellmauthaler, Hannes Strass, Johannes~Peter Wallner, and
  Stefan Woltran.
\newblock Abstract dialectical frameworks revisited.
\newblock In {\em Proceedings of the Twenty-Third international joint
  conference on Artificial Intelligence}, pages 803--809, 2013.

\bibitem{brewka-woltran-10}
Gerhard Brewka and Stefan Woltran.
\newblock Abstract dialectical frameworks.
\newblock In {\em Proceedings of the Twelfth International Conference on
  Principles of Knowledge Representation and Reasoning}, pages 102--111, 2010.

\bibitem{bryant-92}
R.~E. Bryant.
\newblock Symbolic {B}oolean manipulation with ordered binary-decision
  diagrams.
\newblock {\em ACM Computing Surveys}, 24(3):293--318, 1992.

\bibitem{bryant-86}
Randal~E Bryant.
\newblock Graph-based algorithms for {B}oolean function manipulation.
\newblock {\em IEEE Transactions on Computers}, 100(8):677--691, 1986.

\bibitem{caminada-06a}
Martin Caminada.
\newblock On the issue of reinstatement in argumentation.
\newblock Technical Report UU-CS-2006-023, Department of Information and
  Computing Sciences, Utrecht University, May 2006.

\bibitem{caminada-08}
Martin Caminada.
\newblock A gentle introduction to argumentation semantics.
\newblock {\em Lecture material, Summer}, 2008.
\newblock
  \url{http://www.sci.brooklyn.cuny.edu/~parsons/courses/7412-fall-2011/papers/caminada-intro.pdf}.

\bibitem{caminada-carnielli-dunne-06}
Martin~W.A. Caminada, Walter~A. Carnielli, and Paul~E. Dunne.
\newblock Semi-stable semantics.
\newblock {\em Computational Models of Argument (COMMA)}, 144:121--130, 2006.

\bibitem{caminada-carnielli-dunne-11}
Martin~WA Caminada, Walter~A Carnielli, and Paul~E Dunne.
\newblock Semi-stable semantics.
\newblock {\em Journal of Logic and Computation}, 22(5):1207--1254, 2011.

\bibitem{caminada-gabbay-09}
Martin~WA Caminada and Dov~M Gabbay.
\newblock A logical account of formal argumentation.
\newblock {\em Studia Logica}, 93:109--145, 2009.

\bibitem{cayrol-lagasquie-schiex-13}
Claudette Cayrol and Marie-Christine Lagasquie-Schiex.
\newblock Bipolarity in argumentation graphs: Towards a better understanding.
\newblock {\em International Journal of Approximate Reasoning}, 54(7):876--899,
  2013.

\bibitem{chaouiya-remy-ruet-thieffry-04}
Claudine Chaouiya, Elisabeth Remy, Paul Ruet, and Denis Thieffry.
\newblock Qualitative modelling of genetic networks: {F}rom logical regulatory
  graphs to standard {P}etri nets.
\newblock In {\em Applications and Theory of Petri Nets 2004: 25th
  International Conference, ICATPN 2004, Bologna, Italy, June 21--25, 2004.
  Proceedings 25}, pages 137--156. Springer, 2004.

\bibitem{chatain-haar-pauleve-18}
Thomas Chatain, Stefan Haar, and Lo{\"\i}c Paulev{\'e}.
\newblock Boolean networks: beyond generalized asynchronicity.
\newblock In {\em Cellular Automata and Discrete Complex Systems: 24th IFIP WG
  1.5 International Workshop, AUTOMATA 2018, Ghent, Belgium, June 20--22, 2018,
  Proceedings 24}, pages 29--42. Springer, 2018.

\bibitem{cohen-gottifredi-garcia-simari-14}
Andrea Cohen, Sebastian Gottifredi, Alejandro~J Garc{\'\i}a, and Guillermo~R
  Simari.
\newblock A survey of different approaches to support in argumentation systems.
\newblock {\em The Knowledge Engineering Review}, 29(5):513--550, 2014.

\bibitem{comet-et-al-13}
Jean-Paul Comet, Mathilde Noual, Adrien Richard, Julio Aracena, Laurence
  Calzone, Jacques Demongeot, Marcelle Kaufman, Aur{\'e}lien Naldi, El~Houssine
  Snoussi, and Denis Thieffry.
\newblock On circuit functionality in {B}oolean networks.
\newblock {\em Bulletin of Mathematical Biology}, 75:906--919, 2013.

\bibitem{cook-04}
Roy~T Cook.
\newblock Patterns of paradox.
\newblock {\em The Journal of Symbolic Logic}, 69(3):767--774, 2004.

\bibitem{darwiche-marquis-02}
Adnan Darwiche and Pierre Marquis.
\newblock A knowledge compilation map.
\newblock {\em Journal of Artificial Intelligence Research}, 17:229--264, 2002.

\bibitem{de-jong-02}
Hidde de~Jong.
\newblock Modeling and simulation of genetic regulatory systems: a literature
  review.
\newblock {\em Journal of Computational Biology}, 9(1):67--103, 2002.

\bibitem{de-leon-vazquez-jimenez-matadamas-guzman-resendis-antonio-22}
Ugo Avila-Ponce de~Le{\'o}n, Aar{\'o}n V{\'a}zquez-Jim{\'e}nez, Meztli
  Matadamas-Guzm{\'a}n, and Osbaldo Resendis-Antonio.
\newblock Boolean modeling reveals that cyclic attractors in macrophage
  polarization serve as reservoirs of states to balance external perturbations
  from the tumor microenvironment.
\newblock {\em Frontiers in Immunology}, 13:1012730, 2022.

\bibitem{deichmann-schuster-mazat-cornish-bowden-14}
Ute Deichmann, Stefan Schuster, Jean-Pierre Mazat, and Athel Cornish-Bowden.
\newblock Commemorating the 1913 {M}ichaelis--{M}enten paper \emph{{D}ie
  {K}inetik der {I}nvertinwirkung}: three perspectives.
\newblock {\em The FEBS journal}, 281(2):435--463, 2014.

\bibitem{diller-wallner-woltran-15}
Martin Diller, Johannes~Peter Wallner, and Stefan Woltran.
\newblock Reasoning in abstract dialectical frameworks using quantified
  {B}oolean formulas.
\newblock {\em Argument \& Computation}, 6(2):149--177, 2015.

\bibitem{dung-95}
Phan~Minh Dung.
\newblock On the acceptability of arguments and its fundamental role in
  nonmonotonic reasoning, logic programming and $n$-person games.
\newblock {\em Artificial intelligence}, 77(2):321--357, 1995.

\bibitem{dyrkolbotn-12}
Sjur Dyrkolbotn.
\newblock {\em Argumentation, paradox and kernels in directed graphs}.
\newblock PhD thesis, University of Bergen, Norway, 2012.

\bibitem{dyrkolbotn-walicki-14}
Sjur Dyrkolbotn and Micha{\l} Walicki.
\newblock Propositional discourse logic.
\newblock {\em Synthese}, 191:863--899, 2014.

\bibitem{eiter-ianni-krennwallner-09}
Thomas Eiter, Giovambattista Ianni, and Thomas Krennwallner.
\newblock Answer set programming: A primer.
\newblock In {\em Fifth International Reasoning Web Summer School (RW’09)},
  pages 40--110. Springer, 2009.
\newblock LNCS Vol. 5689.

\bibitem{ellmauthaler-23}
Stefan Ellmauthaler.
\newblock Personal communication, September 2023.

\bibitem{ellmauthaler-gaggl-rusovac-wallner-22}
Stefan Ellmauthaler, Sarah~A Gaggl, Dominik Rusovac, and Johannes~P Wallner.
\newblock {ADF-BDD}: An {ADF} solver based on binary decision diagrams.
\newblock In {\em 9th International Conference on Computational Models of
  Argument: COMMA 2022}, pages 355--356, 2022.

\bibitem{ellmauthaler-strass-13}
Stefan Ellmauthaler and Hannes Strass.
\newblock The {DIAMOND} system for argumentation: {P}reliminary report.
\newblock {\em arXiv preprint arXiv:1312.6140}, 2013.

\bibitem{ellmauthaler-strass-14}
Stefan Ellmauthaler and Hannes Strass.
\newblock The {DIAMOND} system for computing with abstract dialectical
  frameworks.
\newblock In {\em COMMA}, pages 233--240, 2014.

\bibitem{fages-soliman-08}
Fran\c{c}ois Fages and Sylvain Soliman.
\newblock Abstract interpretation and types for systems biology.
\newblock {\em Theoretical Computer Science}, 403:52--70, 2008.

\bibitem{nasrollahi-campbell-albert-23}
Fatemeh~Sadat Fatemi~Nasrollahi, Colin Campbell, and R{\'e}ka Albert.
\newblock Predicting cascading extinctions and efficient restoration strategies
  in plant-pollinator networks via generalized positive feedback loops.
\newblock {\em Scientific Reports}, 13(1):902, 2023.

\bibitem{faure-thieffry-09}
Adrien Faur{\'e} and Denis Thieffry.
\newblock Logical modelling of cell cycle control in eukaryotes: a comparative
  study.
\newblock {\em Molecular BioSystems}, 5(12):1569--1581, 2009.

\bibitem{fogelman-soulie-85}
F~Fogelman-Souli\'e.
\newblock Parallel and sequential computation on boolean networks.
\newblock {\em Theoretical Computer Science}, 40:275--300, 1985.

\bibitem{gaucherel-thero-puiseux-bonhomme-17}
C{\'e}dric Gaucherel, H{\'e}lo{\"\i}se Th{\'e}ro, A~Puiseux, and Vincent
  Bonhomme.
\newblock Understand ecosystem regime shifts by modelling ecosystem development
  using {B}oolean networks.
\newblock {\em Ecological Complexity}, 31:104--114, 2017.

\bibitem{ge-qian-09}
Hao Ge and Min Qian.
\newblock Boolean network approach to negative feedback loops of the p53
  pathways: synchronized dynamics and stochastic limit cycles.
\newblock {\em Journal of Computational Biology}, 16(1):119--132, 2009.

\bibitem{gebser-kaminski-kaufmann-schaub-22}
Martin Gebser, Roland Kaminski, Benjamin Kaufmann, and Torsten Schaub.
\newblock {\em Answer Set Solving in Practice}.
\newblock Springer Nature, 2022.

\bibitem{gebser-kaufmann-kaminski-ostrowski-schaub-schneider-11}
Martin Gebser, Benjamin Kaufmann, Roland Kaminski, Max Ostrowski, Torsten
  Schaub, and Marius Schneider.
\newblock Potassco: The {P}otsdam answer set solving collection.
\newblock {\em AI Communications}, 24(2):107--124, 2011.

\bibitem{grandi-lorini-perrussel-15}
Umberto Grandi, Emiliano Lorini, and Laurent Perrussel.
\newblock Propositional opinion diffusion.
\newblock In {\em 14th International Joint Conference on Autonomous Agents and
  Multiagent Systems (AAMAS 2015)}, pages 989--997. International Foundation
  for Autonomous Agents and Multiagent Systems (IFAAMAS), 2015.

\bibitem{grossi-09}
Davide Grossi.
\newblock Doing argumentation theory in modal logic.
\newblock 2009.
\newblock
  \url{https://eprints.illc.uva.nl/id/eprint/354/1/PP-2009-24.text.pdf}.

\bibitem{grossi-12}
Davide Grossi.
\newblock Fixpoints and iterated updates in abstract argumentation.
\newblock In {\em Thirteenth International Conference on the Principles of
  Knowledge Representation and Reasoning}, 2012.

\bibitem{grossi-modgil-19}
Davide Grossi and Sanjay Modgil.
\newblock On the graded acceptability of arguments in abstract and instantiated
  argumentation.
\newblock {\em Artificial Intelligence}, 275:138--173, 2019.

\bibitem{harris-sawhill-wuensche-kauffman-02}
Stephen~E Harris, Bruce~K Sawhill, Andrew Wuensche, and Stuart Kauffman.
\newblock A model of transcriptional regulatory networks based on biases in the
  observed regulation rules.
\newblock {\em Complexity}, 7(4):23--40, 2002.

\bibitem{heyninck-knorr-leite-24}
Jesse Heyninck, Matthias Knorr, and Jo{\~a}o Leite.
\newblock Abstract dialectical frameworks are {B}oolean networks (full
  version).
\newblock {\em arXiv preprint arXiv:2407.02055}, 2024.

\bibitem{huth-ryan-04}
Michael R.~A. Huth and Mark~D. Ryan.
\newblock {\em Logic in Computer Science: Modelling and Reasoning about
  Systems}.
\newblock Cambridge University Press, 2nd edition, 2004.

\bibitem{irons-06}
David~James Irons.
\newblock Improving the efficiency of attractor cycle identification in
  {B}oolean networks.
\newblock {\em Physica D: Nonlinear Phenomena}, 217(1):7--21, 2006.

\bibitem{jarrah-raposa-laubenbacher-07}
Abdul~Salam Jarrah, Blessilda Raposa, and Reinhard Laubenbacher.
\newblock Nested canalyzing, unate cascade, and polynomial functions.
\newblock {\em Physica D: Nonlinear Phenomena}, 233(2):167--174, 2007.

\bibitem{kauffman-71}
Stuart Kauffman.
\newblock Gene regulation networks: A theory for their global structure and
  behaviors.
\newblock In {\em Current topics in developmental biology}, volume~6,
  chapter~5, pages 145--182. Elsevier, 1971.

\bibitem{kauffman-74}
Stuart Kauffman.
\newblock The large scale structure and dynamics of gene control circuits: an
  ensemble approach.
\newblock {\em Journal of Theoretical Biology}, 44(1):167--190, 1974.

\bibitem{kauffman-69}
Stuart~A Kauffman.
\newblock Metabolic stability and epigenesis in randomly constructed genetic
  nets.
\newblock {\em Journal of Theoretical Biology}, 22(3):437--467, 1969.

\bibitem{kauffman-72}
Stuart~A Kauffman.
\newblock The organization of cellular genetic control systems.
\newblock {\em Some Mathematical Questions in Biology}, pages 63--116, 1972.

\bibitem{kauffman-93}
Stuart~A Kauffman.
\newblock {\em The origins of order: Self-organization and selection in
  evolution}.
\newblock Oxford University Press, USA, 1993.

\bibitem{klarner-14}
Hannes Klarner.
\newblock {\em Contributions to the analysis of qualitative models of
  regulatory networks}.
\newblock PhD thesis, Freie Universit\"at Berlin, November 2014.

\bibitem{klarner-bockmayr-siebert-14}
Hannes Klarner, Alexander Bockmayr, and Heike Siebert.
\newblock Computing symbolic steady states of {B}oolean networks.
\newblock In {\em Cellular Automata: 11th International Conference on Cellular
  Automata for Research and Industry, ACRI 2014, Krakow, Poland, September
  22--25, 2014. Proceedings 11}, pages 561--570. Springer, 2014.

\bibitem{klarner-bockmayr-siebert-15}
Hannes Klarner, Alexander Bockmayr, and Heike Siebert.
\newblock Computing maximal and minimal trap spaces of {B}oolean networks.
\newblock {\em Natural Computing}, 14:535--544, 2015.

\bibitem{klarner-siebert-15}
Hannes Klarner and Heike Siebert.
\newblock Approximating attractors of {B}oolean networks by iterative {CTL}
  model checking.
\newblock {\em Frontiers in Bioengineering and Biotechnology}, 3:130, 2015.

\bibitem{klarner-streck-siebert-17}
Hannes Klarner, Adam Streck, and Heike Siebert.
\newblock {P}y{B}ool{N}et: a python package for the generation, analysis and
  visualization of boolean networks.
\newblock {\em Bioinformatics}, 33(5):770--772, 2017.

\bibitem{kleene-52}
Stephen~Cole Kleene.
\newblock {\em Introduction to Metamathematics}.
\newblock Wolters-Noordhoff, North-Holland and American Elsevier, 1952.

\bibitem{neumann-lara-71}
Víctor~Neumann Lara.
\newblock Seminúcleos de una digráfica.
\newblock Technical Report Volumen 11, Anales del Instituto de Matemáticas,
  Universidad Nacional Autónoma de México, 1971.

\bibitem{le-novere-15}
Nicolas Le~Novere.
\newblock Quantitative and logic modelling of molecular and gene networks.
\newblock {\em Nature Reviews Genetics}, 16(3):146--158, 2015.

\bibitem{li-adeyeye-murrugarra-aguilar-laubenbacher-13}
Yuan Li, John~O Adeyeye, David Murrugarra, Boris Aguilar, and Reinhard
  Laubenbacher.
\newblock Boolean nested canalizing functions: {A} comprehensive analysis.
\newblock {\em Theoretical Computer Science}, 481:24--36, 2013.

\bibitem{lifschitz-19}
Vladimir Lifschitz.
\newblock {\em Answer Set Programming}, volume~3.
\newblock Springer Heidelberg, 2019.

\bibitem{linsbichler-maratea-niskanen-wallner-woltran-18}
Thomas Linsbichler, Marco Maratea, Andreas Niskanen, Johannes~P Wallner, and
  Stefan Woltran.
\newblock Novel algorithms for abstract dialectical frameworks based on
  complexity analysis of subclasses and sat solving.
\newblock In {\em Proceedings of the 27. International Joint Conference on
  Artificial Intelligence, Stockholm, 13-19 July 2018/Lang, Jerome}, volume~27.
  ijcai. org, 2018.

\bibitem{mcnaughton-61}
Robert McNaughton.
\newblock Unate truth functions.
\newblock {\em IRE Transactions on Electronic Computers}, (1):1--6, 1961.

\bibitem{meinel-theobald-1998}
Christoph Meinel and Thorsten Theobald.
\newblock {\em Algorithms and Data Structures in {VLSI} Design: {OBDD}
  Foundations and Applications}.
\newblock Springer, 1998.

\bibitem{mercier-sperber-11}
Hugo Mercier and Dan Sperber.
\newblock Why do humans reason? arguments for an argumentative theory.
\newblock {\em Behavioral and brain sciences}, 34(2):57--74, 2011.

\bibitem{mercier-sperber-17}
Hugo Mercier and Dan Sperber.
\newblock {\em The Enigma of Reason}.
\newblock Harvard University Press, 2017.

\bibitem{michaelis-menten-13}
Leonor Michaelis and Maud~L Menten.
\newblock Die {K}inetik der {I}nvertinwirkung.
\newblock {\em Biochem. z}, 49(333--369):352, 1913.

\bibitem{montagud-et-al-22}
Arnau Montagud, Jonas B{\'e}al, Luis Tobalina, Pauline Traynard, Vigneshwari
  Subramanian, Bence Szalai, R{\'o}bert Alf{\"o}ldi, L{\'a}szl{\'o} Pusk{\'a}s,
  Alfonso Valencia, Emmanuel Barillot, Julio Saez-Rodriguez, and Laurence
  Calzone.
\newblock Patient-specific {B}oolean models of signalling networks guide
  personalised treatments.
\newblock {\em Elife}, 11:e72626, 2022.

\bibitem{moon-lee-pauleve-23}
Kyungduk Moon, Kangbok Lee, and Loïc Paulevé.
\newblock Computational complexity of minimal trap spaces in {B}oolean
  networks, 2023.

\bibitem{mori-mochizuki-17}
Fumito Mori and Atsushi Mochizuki.
\newblock Expected number of fixed points in {B}oolean networks with arbitrary
  topology.
\newblock {\em Physical Review Letters}, 119(2):028301, 2017.

\bibitem{munoz-carrillo-azpeitia-rosenblueth-18}
Stalin Munoz, Miguel Carrillo, Eugenio Azpeitia, and David~A Rosenblueth.
\newblock Griffin: a tool for symbolic inference of synchronous {B}oolean
  molecular networks.
\newblock {\em Frontiers in Genetics}, 9:39, 2018.

\bibitem{naldi-18}
Aur{\'e}lien Naldi.
\newblock Bio{LQM}: a {J}ava toolkit for the manipulation and conversion of
  logical qualitative models of biological networks.
\newblock {\em Frontiers in Physiology}, 9:382371, 2018.

\bibitem{naldi-et-al-15}
Aur{\'e}lien Naldi, Pedro~T Monteiro, Christoph M{\"u}ssel, Consortium for
  Logical~Models, Tools, Hans~A Kestler, Denis Thieffry, Ioannis Xenarios,
  Julio Saez-Rodriguez, Tomas Helikar, and Claudine Chaouiya.
\newblock Cooperative development of logical modelling standards and tools with
  {C}o{L}o{M}o{T}o.
\newblock {\em Bioinformatics}, 31(7):1154--1159, 2015.

\bibitem{naldi-remi-thieffry-chaouiya-11}
Aur{\'e}lien Naldi, Elisabeth Remy, Denis Thieffry, and Claudine Chaouiya.
\newblock Dynamically consistent reduction of logical regulatory graphs.
\newblock {\em Theoretical Computer Science}, 412(21):2207--2218, 2011.

\bibitem{naldi-thieffry-chaouiya-07}
Aur{\'e}lien Naldi, Denis Thieffry, and Claudine Chaouiya.
\newblock Decision diagrams for the representation and analysis of logical
  models of genetic networks.
\newblock In {\em Computational Methods in Systems Biology: International
  Conference CMSB 2007, Edinburgh, Scotland, September 20-21, 2007.
  Proceedings}, pages 233--247. Springer, 2007.

\bibitem{von-neumann-66}
John~von Neumann.
\newblock Theory of self-reproducing automata.
\newblock {\em Edited by Arthur W. Burks}, 1966.

\bibitem{nouioua-risch-10}
Farid Nouioua and Vincent Risch.
\newblock Bipolar argumentation frameworks with specialized supports.
\newblock In {\em 2010 22nd IEEE International Conference on Tools with
  Artificial Intelligence}, volume~1, pages 215--218. IEEE, 2010.

\bibitem{o-donell-14}
Ryan O'Donnell.
\newblock {\em Analysis of {B}oolean Functions}.
\newblock Cambridge University Press, 2014.

\bibitem{ordaz-arias-diaz-alvarez-zuniga-martinez-sanchez-balderas-martinez-22}
Manuel~Azaid Ordaz-Arias, Laura D{\'\i}az-Alvarez, Joaqu{\'\i}n Z{\'u}{\~n}iga,
  Mariana~Esther Martinez-S{\'a}nchez, and Yalbi~Itzel Balderas-Mart{\'\i}nez.
\newblock Cyclic attractors are critical for macrophage differentiation,
  heterogeneity, and plasticity.
\newblock {\em Frontiers in Molecular Biosciences}, 9:807228, 2022.

\bibitem{papadimitriou-94}
Christos~H. Papadimitriou.
\newblock {\em Computational Complexity}.
\newblock Addison-Wesley, 1994.

\bibitem{pastva-safranek-benes-brim-henzinger-23}
Samuel Pastva, David Safranek, Nikola Benes, Lubos Brim, and Thomas Henzinger.
\newblock Repository of logically consistent real-world {B}oolean network
  models.
\newblock {\em bioRxiv}, pages 2023--06, 2023.

\bibitem{pauleve-kolcak-chatain-haar-20}
Lo{\"\i}c Paulev{\'e}, Juri Kol{\v{c}}{\'a}k, Thomas Chatain, and Stefan Haar.
\newblock Reconciling qualitative, abstract, and scalable modeling of
  biological networks.
\newblock {\em Nature Communications}, 11(1):4256, 2020.

\bibitem{pauleve-richard-12}
Lo{\"\i}c Paulev{\'e} and Adrien Richard.
\newblock Static analysis of {B}oolean networks based on interaction graphs:
  {A} survey.
\newblock {\em Electronic Notes in Theoretical Computer Science}, 284:93--104,
  2012.

\bibitem{pavlidis-12}
Theodosios Pavlidis.
\newblock {\em Biological oscillators: their mathematical analysis}.
\newblock Elsevier, 2012.

\bibitem{peterson-77}
James~L Peterson.
\newblock Petri nets.
\newblock {\em ACM Computing Surveys}, 9(3):223--252, 1977.

\bibitem{poindron-21}
Alexis Poindron.
\newblock A general model of binary opinions updating.
\newblock {\em Mathematical Social Sciences}, 109:52--76, 2021.

\bibitem{polberg-oren-14}
Sylwia Polberg and Nir Oren.
\newblock Revisiting support in abstract argumentation systems.
\newblock In {\em COMMA}, pages 369--376, 2014.

\bibitem{pollock-87}
John~L Pollock.
\newblock Defeasible reasoning.
\newblock {\em Cognitive Science}, 11(4):481--518, 1987.

\bibitem{pollock-91}
John~L Pollock.
\newblock A theory of defeasible reasoning.
\newblock {\em International Journal of Intelligent Systems}, 6(1):33--54,
  1991.

\bibitem{pollock-92}
John~L Pollock.
\newblock How to reason defeasibly.
\newblock {\em Artificial Intelligence}, 57(1):1--42, 1992.

\bibitem{prakken-21}
Henry Prakken.
\newblock Logical models of legal argumentation.
\newblock In Markus Knauff and Wolfgang Spohn, editors, {\em The Handbook of
  Rationality}, chapter 11.4, pages 669--677. The MIT Press, 2021.

\bibitem{prakken-vreeswijk-02}
Henry Prakken and Gerard Vreeswijk.
\newblock Logics for defeasible argumentation.
\newblock {\em Handbook of philosophical logic}, pages 219--318, 2002.

\bibitem{proietti-yuste-ginel-21}
Carlo Proietti and Antonio Yuste-Ginel.
\newblock Dynamic epistemic logics for abstract argumentation.
\newblock {\em Synthese}, 199(3):8641--8700, 2021.

\bibitem{quine-52}
Willard~V Quine.
\newblock The problem of simplifying truth functions.
\newblock {\em The American mathematical monthly}, 59(8):521--531, 1952.

\bibitem{raeymaekers-02}
Luc Raeymaekers.
\newblock Dynamics of {B}oolean networks controlled by biologically meaningful
  functions.
\newblock {\em Journal of Theoretical Biology}, 218(3):331--341, 2002.

\bibitem{remy-ruet-08}
Elisabeth Remy and Paul Ruet.
\newblock From minimal signed circuits to the dynamics of {B}oolean regulatory
  networks.
\newblock {\em Bioinformatics}, 24(16):i220--i226, 2008.

\bibitem{remy-ruet-thieffry-08}
{\'E}lisabeth Remy, Paul Ruet, and Denis Thieffry.
\newblock Graphic requirements for multistability and attractive cycles in a
  {B}oolean dynamical framework.
\newblock {\em Advances in Applied Mathematics}, 41(3):335--350, 2008.

\bibitem{richard-10}
Adrien Richard.
\newblock Negative circuits and sustained oscillations in asynchronous automata
  networks.
\newblock {\em Advances in Applied Mathematics}, 44(4):378--392, 2010.

\bibitem{richard-19}
Adrien Richard.
\newblock Positive and negative cycles in {B}oolean networks.
\newblock {\em Journal of Theoretical Biology}, 463:67--76, 2019.

\bibitem{richard-comet-07}
Adrien Richard and Jean-Paul Comet.
\newblock Necessary conditions for multistationarity in discrete dynamical
  systems.
\newblock {\em Discrete Applied Mathematics}, 155(18):2403--2413, 2007.

\bibitem{richard-rossignol-comet-bernot-guespin-michel-merieau-12}
Adrien Richard, Gaelle Rossignol, Jean-Paul Comet, Gilles Bernot, Jannine
  Guespin-Michel, and Annabelle Merieau.
\newblock Boolean models of biosurfactants production in \emph{{P}seudomonas
  fluorescens}.
\newblock {\em PLoS One}, 7(1):e24651, 2012.

\bibitem{richard-ruet-13}
Adrien Richard and Paul Ruet.
\newblock From kernels in directed graphs to fixed points and negative cycles
  in {B}oolean networks.
\newblock {\em Discrete Applied Mathematics}, 161(7-8):1106--1117, 2013.

\bibitem{robert-78}
Fran{\c{c}}ois Robert.
\newblock Th{\'e}or{\`e}mes de {P}erron--{F}robenius et {S}tein--{R}osenberg
  bool{\'e}ens.
\newblock {\em Linear Algebra and its Applications}, 19(3):237--250, 1978.

\bibitem{robert-80}
Fran{\c{c}}ois Robert.
\newblock Iterations sur des ensembles finis et automates cellulaires
  contractants.
\newblock {\em Linear Algebra and its applications}, 29:393--412, 1980.

\bibitem{robert-86}
Fran{\c{c}}ois Robert.
\newblock {\em Discrete Iterations: a Metric Study}, volume~6.
\newblock Springer Series in Computational Mathematics, Vol. 6, 1986.
\newblock (transl.).

\bibitem{rozum-campbell-newby-nasrollahi-albert-23}
Jordan~C Rozum, Colin Campbell, Eli Newby, Fatemeh Sadat~Fatemi Nasrollahi, and
  Reka Albert.
\newblock Boolean networks as predictive models of emergent biological
  behaviors.
\newblock {\em arXiv preprint arXiv:2310.12901}, 2023.

\bibitem{saadatpour-albert-reluga-13}
Assieh Saadatpour, R{\'e}ka Albert, and Timothy~C Reluga.
\newblock A reduction method for {B}oolean network models proven to conserve
  attractors.
\newblock {\em SIAM Journal on Applied Dynamical Systems}, 12(4):1997--2011,
  2013.

\bibitem{samaga-klamt-13}
Regina Samaga and Steffen Klamt.
\newblock Modeling approaches for qualitative and semi-quantitative analysis of
  cellular signaling networks.
\newblock {\em Cell communication and signaling}, 11(1):1--19, 2013.

\bibitem{schwab-kuhlwein-ikonomi-kuhl-kestler-20}
Julian~D Schwab, Silke~D K{\"u}hlwein, Nensi Ikonomi, Michael K{\"u}hl, and
  Hans~A Kestler.
\newblock Concepts in {B}oolean network modeling: {W}hat do they all mean?
\newblock {\em Computational and Structural Biotechnology Journal},
  18:571--582, 2020.

\bibitem{schwarzentruber-vesic-rienstra-12}
Fran{\c{c}}ois Schwarzentruber, Srdjan Vesic, and Tjitze Rienstra.
\newblock Building an epistemic logic for argumentation.
\newblock In {\em European Workshop on Logics in Artificial Intelligence},
  pages 359--371. Springer, 2012.

\bibitem{siebert-11}
Heike Siebert.
\newblock Analysis of discrete bioregulatory networks using symbolic steady
  states.
\newblock {\em Bulletin of mathematical biology}, 73:873--898, 2011.

\bibitem{strass-13}
Hannes Strass.
\newblock Approximating operators and semantics for abstract dialectical
  frameworks.
\newblock {\em Artificial Intelligence}, 205:39--70, 2013.

\bibitem{strass-ellmauthaler-17}
Hannes Strass and Stefan Ellmauthaler.
\newblock {goDIAMOND} 0.6.6 {ICCMA} 2017 system description, 2017.
\newblock \url{https://iccl.inf.tu-dresden.de/w/images/f/f2/GoDIAMOND.pdf}.

\bibitem{strass-wallner-15}
Hannes Strass and Johannes~Peter Wallner.
\newblock Analyzing the computational complexity of abstract dialectical
  frameworks via approximation fixpoint theory.
\newblock {\em Artificial Intelligence}, 226:34--74, 2015.

\bibitem{strzemecki-92}
Tadeusz Strzemecki.
\newblock Polynomial-time algorithms for generation of prime implicants.
\newblock {\em Journal of Complexity}, 8(1):37--63, 1992.

\bibitem{su-pang-paul-19}
Cui Su, Jun Pang, and Soumya Paul.
\newblock Towards optimal decomposition of {B}oolean networks.
\newblock {\em IEEE/ACM Transactions on Computational Biology and
  Bioinformatics}, 18(6):2167--2176, 2019.

\bibitem{subbaroyan-margin-samal-22}
Ajay Subbaroyan, Olivier~C Martin, and Areejit Samal.
\newblock Minimum complexity drives regulatory logic in {B}oolean models of
  living systems.
\newblock {\em PNAS nexus}, 1(1):1--12, 2022.

\bibitem{tarski-55}
Alfred Tarski.
\newblock A lattice-theoretical fixpoint theorem and its applications.
\newblock {\em Pacific Journal of Mathematics}, 5:285--310, 1955.

\bibitem{thomas-81}
Ren{\'e} Thomas.
\newblock On the relation between the logical structure of systems and their
  ability to generate multiple steady states or sustained oscillations.
\newblock In {\em Numerical Methods in the Study of Critical Phenomena:
  Proceedings of a Colloquium, Carry-le-Rouet, France, June 2--4, 1980}, pages
  180--193. Springer, 1981.

\bibitem{thomas-d-ari-90}
Ren\'e Thomas and Richard D'Ari.
\newblock {\em Biological Feedback}.
\newblock CRC Press, 1990.

\bibitem{toulmin-58}
Stephen~E Toulmin.
\newblock {\em The Uses of Argument}.
\newblock Cambridge University Press, 1958.

\bibitem{trinh-benhamou-pastva-soliman-24}
Giang Trinh, Belaid Benhamou, Samuel Pastva, and Sylvain Soliman.
\newblock Scalable enumeration of trap spaces in {B}oolean networks via
  {A}nswer {S}et {P}rogramming.
\newblock In {\em Proceedings of the AAAI Conference on Artificial
  Intelligence}, volume~38, pages 10714--10722, 2024.

\bibitem{trinh-benhamou-pauleve-24}
Van-Giang Trinh, Belaid Benhamou, and Loïc Paulevé.
\newblock \texttt{mpbn}: a simple tool for efficient edition and analysis of
  elementary properties of {B}oolean networks, 2024.

\bibitem{trinh-benhamou-soliman-23}
Van-Giang Trinh, Belaid Benhamou, and Sylvain Soliman.
\newblock Trap spaces of {B}oolean networks are conflict-free siphons of their
  {P}etri net encoding.
\newblock {\em Theoretical Computer Science}, 971:114073, 2023.

\bibitem{trinh-hiraishi-benhamou-22}
Van-Giang Trinh, Kunihiko Hiraishi, and Belaid Benhamou.
\newblock Computing attractors of large-scale asynchronous {B}oolean networks
  using minimal trap spaces.
\newblock In {\em Proceedings of the 13th ACM International Conference on
  Bioinformatics, Computational Biology and Health Informatics}, pages 1--10,
  2022.

\bibitem{veliz-cuba-11}
Alan Veliz-Cuba.
\newblock Reduction of {B}oolean network models.
\newblock {\em Journal of Theoretical Biology}, 289:167--172, 2011.

\bibitem{verheij-96}
Bart Verheij.
\newblock Two approaches to dialectical argumentation: admissible sets and
  argumentation stages.
\newblock {\em Proc. NAIC}, 96:357--368, 1996.

\bibitem{waddington-42}
C~H Waddington.
\newblock Canalization of development and the inheritance of acquired
  characters.
\newblock {\em Nature}, 150(3811):563--565, 1942.

\bibitem{wang-saadatpour-albert-12}
Rui-Sheng Wang, Assieh Saadatpour, and Reka Albert.
\newblock Boolean modeling in systems biology: an overview of methodology and
  applications.
\newblock {\em Physical Biology}, 9(5):055001, 2012.

\bibitem{wegener-00}
Ingo Wegener.
\newblock {\em Branching Programs and Binary Decision Diagrams: Theory and
  Applications}.
\newblock SIAM Monographs on Discrete Mathematics and Applications, 2000.

\bibitem{woods-21}
John Woods.
\newblock Reasoning and argumentation.
\newblock In Markus Knauff and Wolfgang Spohn, editors, {\em The Handbook of
  Rationality}, chapter 5.6, pages 367--377. The MIT Press, 2021.

\bibitem{zanudo-albert-13}
Jorge~GT Za{\~n}udo and R{\'e}ka Albert.
\newblock An effective network reduction approach to find the dynamical
  repertoire of discrete dynamic networks.
\newblock {\em Chaos: An Interdisciplinary Journal of Nonlinear Science},
  23(2), 2013.

\end{thebibliography}
\end{document}